\documentclass[a4paper, oneside, reqno]{amsart}
\usepackage{siunitx}
\usepackage{dsfont}
\usepackage{graphicx}
\usepackage{multirow}
\usepackage{amsmath,bbm, amssymb,amsfonts}
\usepackage{amsthm, thmtools}
\usepackage{mathrsfs}
\usepackage{mathtools}
\usepackage[title]{appendix}
\usepackage[dvipsnames]{xcolor}
\usepackage{textcomp}
\usepackage{manyfoot}
\usepackage{booktabs}
\usepackage{algorithm}
\usepackage{algorithmicx}
\usepackage{algpseudocode}
\usepackage{listings}
\usepackage{tabularx}
\usepackage{varioref}
\usepackage{hyperref}
\usepackage{cleveref}
\usepackage[colorinlistoftodos,prependcaption,textsize=tiny]{todonotes}

\newcommand{\per}{\mathrm{per}}
\newcommand{\fl}[1]{\mathrm{fl}{\left[#1 \right]}}

\newcommand{\norm}[1]{\left\lVert#1\right\rVert}
\newcommand{\abs}[1]{\left\lvert#1\right\rvert}

\declaretheorem[name=Theorem]{theorem}
\declaretheorem[name=Lemma]{lemma}
\theoremstyle{remark}
\declaretheorem[name=Remark]{remark}

\title[Numerical Method for the Left Tail of Sums of Independent RVs]{A Fast and Accurate Numerical Method for the Left Tail of Sums of Independent Random Variables}

\author{Nadhir Ben Rached}
\address{School of Mathematics, University of Leeds, United Kingdom}
\email{n.benrached@leeds.ac.uk}

\author{H{\aa}kon Hoel}
\address{Department of Mathematics, University of Oslo, Norway}
\email{haakonah@math.uio.no}

\author{Johannes Vincent Meo$^\star$}
\address{Department of Mathematics, University of Oslo, Norway}
\email{johannvm@math.uio.no}

\thanks{$^\star$Corresponding author: Johannes Vincent Meo (johannvm@math.uio.no)}


\begin{document}
\begin{abstract}
We present a flexible, deterministic numerical method for computing left-tail rare events of sums of non-negative, independent random variables. The method is based on iterative numerical integration of linear convolutions by means of Newtons--Cotes rules. The periodicity properties of convoluted densities combined with the Trapezoidal rule are exploited to produce a robust and efficient method, and the method is flexible in the sense that it can be applied to all kinds of non-negative continuous RVs. We present an error analysis and study the benefits of utilizing Newton--Cotes rules versus the fast Fourier transform (FFT) for numerical integration, showing that although there can be efficiency-benefits to using FFT, Newton--Cotes rules tend to preserve the relative error better, and indeed do so at an acceptable computational cost.
Numerical studies on problems with both known and unknown rare-event probabilities showcase the method's performance and support our theoretical findings.  
\end{abstract}
\maketitle

\section{Introduction}

For a sequence of non-negative and independent continuous random variables (RVs) $X_1, X_2, \ldots, X_n$, we seek to estimate the probability of failure
\[
\alpha = \mathbb{P} \left (\sum_{i=1}^n X_i <\gamma \right )    
\]
with low relative error for small values of $\gamma>0$. Left-tail rare-event problems of this kind are used for instance to estimate the outage probability in wireless communications, as is described further in the next paragraph. Our quite straightforward deterministic method approximates the density of $\sum_{i=1}^n X_i$ through iterative numerical convolution, and we show both theoretically and in numerical experiments that this approach is robust and that it performs very well in terms of accuracy and efficiency. 

The main inspiration for our approach is Keich~\cite{keich2005sfft,wilson2016accurate}, where a similar approach is studied for approximating the density of a sum of independent and identically distributed (i.i.d.) discrete RVs. The approaches are similar in how rounding errors propagate in computations of linear convolutions and both can be combined with the fast Fourier transform to speed up computations of linear convolutions, at the price of introducing more rounding errors. 
A notable difference is that for continuous RVs, the periodicity of convolutions of densities can be utilized to produce high-order convergence rates in the numerical integration. It is not clear if periodicity can be exploited to improve tractability also for discrete RVs.

Right and left tails of sums of RVs have gained significant attention in the literature due to their broad range of applications. In financial engineering, for example, the value at risk for a portfolio based on multiple assets can be represented as the cumulative distribution function (CDF) of sums of RVs \cite{Asmussen16}. In the performance analysis of wireless communication systems, the outage probability/capacity can be expressed as the CDF of sums of RVs corresponding to either the fading channel envelopes or channel gains \cite{7328688,1275712}. Further applications extend to insurance risk and queueing systems. Within the Cramer-Lundberg model, the total sum of claims is modeled by a random sum of independent RVs, and the ruin probability is defined as the probability that this sum exceeds a large threshold \cite{asmussen2007stochastic}.

Generally, a closed-form expression of the CDF of sums of RVs does not exist for most of the distributions. This is for instance the case
for the Log-Normal distribution which has attracted a substantial
interest
\cite{Asmussen16,FURMAN2020120,10.3150/14-BEJ665,7328688,1275712,9014029,9149127,9367405,ben2023state,asmussen2007stochastic}. Although
several closed-form approximations have been devised, their accuracy
is not guaranteed and may degrade for specific parameter choices
\cite{1275712,9149127,9367405,1578407,6503436,1369233,4781943,8737752,7145970,poly}. Furthermore,
these closed-form approximations are not generic, as they are
generally tailored to specific distributions. Efficient numerical
methods have been proposed in the literature to approximate the
distributions of sums of RVs
\cite{5282371,7378408,5425708,4471924,FURMAN2020120,757200,5464256,4917505}. For
instance, in \cite{5282371}, Smolyak's algorithm, belonging to the
family of numerical integration methods on sparse grids, has been
developed for the accurate analysis of correlated Log-Normal power
sums. Convenient numerical methods for Log-Normal characteristic
functions have also been proposed, as seen in
\cite{4471924,5425708,5464256}. In \cite{7378408}, the authors used a
saddle point approximation to evaluate the outage probability of
wireless cellular networks. However, this method assumes the existence
of the cumulant generating function, a requirement that is not met by
many practical distributions, including the Log-Normal distribution. A
general numerical approach, presented in \cite{4917505}, has also been
developed for computing wireless outages. Similar to \cite{7378408},
this approach is general, provided that the moment generating function is known.

Monte Carlo (MC) methods are versatile tools employed to provide
approximations of the CDF or
complementary CDF of sums of RVs. However, it is widely recognized
that naive MC simulations are computationally expensive, especially
when addressing the right and left tails of sum distributions
\cite{kroese2011handbook}. To mitigate this computational
inefficiency, various efficient variance reduction techniques have
been proposed. While much of the existing research has concentrated on
the right tail of the sum distribution
\cite{asmussen2007stochastic,6376313,7857009,ben2018generalization,juneja2002simulating,asmussen2011efficient,asmussen2006improved,asmussen2012error,asmussen1997simulation},
the left tail region, which is the primary focus of this work, has
only recently gained attention. In \cite{Asmussen16}, the authors utilized the exponential twisting
technique, a well-known importance sampling scheme, to efficiently
estimate the left tail of sums of i.i.d. Log-Normal RVs. Additionally, in \cite{7328688},
two generic importance sampling schemes based on the hazard rate
twisting technique of \cite{juneja2002simulating} were proposed to
estimate the tail of the CDF of independent RVs. These algorithms have
proven to be efficient for a wide range of well-known distributions
within the context of wireless communication systems. An efficient
importance sampling scheme has also been developed for the left tail
of correlated Log-Normals \cite{10.3150/14-BEJ665}. This scheme was
further enhanced in \cite{AlouiniBenRachedKammounTempone}, where
importance sampling and control variates were combined to achieve a
further reduction in variance. Recently, state-dependent importance
sampling has been proposed using a stochastic optimal control
formulation \cite{ben2023state}. This generic approach has been found
to be efficient when considering rare event quantities that take the form of an expected value of some functions applied to sums of independent RVs. 

The rest of this paper is organized as follows. Section~\ref{sec:numericalMethod} describes the numerical method for estimating left-tail rare events. Section~\ref{sec:theory} presents error and cost analysis of the 
rare-event estimation method both when using Newton--Cotes rules and FFT numerical integrators for discrete convolution. Section~\ref{sec:numerics} studies the method numerically on a collection of rare-event problems, and compares its performance to the 
saddle-point method for sums of Log-Normal RVs. Section~\ref{sec:conclusion} summarizes our findings and discusses possible future extensions.

\section{The numerical method} \label{sec:numericalMethod}
In this section we construct an iterative Newton-Cotes quadrature method for approximating linear convolutions of probability densities, and ultimately the probability of failure. For the sake of streamlining the exposition, we will restrict ourselves to the setting with i.i.d.~RVs, but an extension to independent (and not identically distributed) RVs is exemplified in Section~\ref{sec:Lognormal}. Let $f:[0,\infty) \to [0,\infty)$ denote the probability density function (PDF) of the i.i.d. RVs $X_1, \ldots, X_n$. We consider the problem of estimating 
the rare event
\[
\alpha = \mathbb{P} \left (\sum_{i=1}^n X_i <\gamma \right )   = \int_0^\gamma f_{S_n}(x) \, dx
\]
for small values of $\gamma>0$, where $f_{S_n}$ is the PDF of $S_n=\sum_{i=1}^n X_i$. For simplicity, we will make the assumption that $f(0):= \lim_{x \downarrow 0} f(x) = 0$ throughout this section, and describe the extension to settings when $f(0)\neq 0$) in 
equation~\eqref{eq:discLinConv2}. 

\subsection{Numerical integration} \label{sec:numericalConvolution}

The first step in our deterministic approach for estimating $\alpha$ by discrete linear convolution is to observe that due to the independence and identical distribution of $X_1,\ldots,X_n$, the PDF of $S_n := \sum_{i=1}^n X_i$
is equal to the $n-$fold linear convolution of the density $f$:
\[
f_{S_n} = f^{*n} := \text{n-fold convolution of } f.   
\]
In particular, we have that
\[
  f_{S_2}(x) = f*f(x) = \int_0^x f(y) f(x-y) \, dy,
\]
and for any integers $k,\ell,n \ge 1$ such that $n = k + \ell$, it holds that
\begin{equation}\label{eq:nConv}
  f_{S_n}(x) = f^{*k}*f^{*\ell}(x) = \int_0^x f^{*k}(y) f^{*\ell}(x-y) \, dy. 
\end{equation}
This means that the probability of failure can be expressed as  
\[
\alpha = \mathbb{P} \left (\sum_{i=1}^n X_i <\gamma \right ) = \int_0^\gamma f^{*n}(x) \, dx,   
\]
and that a good approximation of $\alpha$ can be obtained from a good approximation of $f^{*n}$.\\
The second step is to approximate $f^{*n}$ on $[0, \gamma]$ by
numerical integration. Since $f(0) = 0$, any integrand 
of the form $g(y;x) = f^{*k}(y) f^{*\ell}(x-y)$ is periodic on
$[0,x]$ for any $x \in [0,\gamma]$, as
\[
  g(0;x) = f^{*k}(0) f^{*\ell}(x)=0, \quad \text{and}
  \quad g(x;x) = f^{*k}(x) f^{*\ell}(0) = 0. 
\]
Conveniently, the trapezoidal rule yields a high order 
of convergence for sufficiently smooth periodic 
integrands, and it is therefore a suitable quadrature rule for linear convolution of~\eqref{eq:nConv}, as described through the following steps:
Consider the uniform mesh
\[
  x_j = j h \quad j =0,1,\ldots,N \quad \text{with} \quad h =\frac{\gamma}{N},
\]
and the discrete function
\[
\bar f(x_j) := f(x_j) \qquad j =0,1,\ldots, N.
\]
Let the discrete approximation of $f^{*2}$ be given by the trapezoidal rule/discrete linear convolution
\begin{equation}\label{eq:fBar2}
\bar f^{\circledast2}(x_k) = h \sum_{j=0}^{k} \bar f(x_j) \bar f(x_{k} - x_j) \qquad k = 0,\ldots, N. 
\end{equation}
The operator notation $\circledast$
represents discrete linear convolution scaled by the step-size factor $h$.
It is introduced to distinguish discrete linear convolution from 
continuous-space linear convolution, which we denoted by $\ast$. 
To define higher-order discrete convolutions we proceed as follows: for any two $\mathbb{R}^{N+1}$-vectors 
$\bar g_1 = (\bar g_1(x_0), \ldots, \bar g_1(x_{N}))$ and $\bar g_2 = (\bar g_2(x_0), \ldots, \bar g_2(x_{N}))$, let  
\[
(\bar g_1 \circledast \bar g_2 )(x_k) := h \sum_{j=0}^{k} \bar g_1(x_j) \bar g_2(x_{k} - x_j) \,.
\]
It then holds that $\circledast$ is an associative operation, as for 
any $\bar g_1, \bar g_2, \bar g_3 \in \mathbb{R}^{N+1}$,
\[
\Big(\bar g_1 \circledast (\bar g_2 \circledast \bar g_3) \Big)(x_k) = 
h^2 \sum_{j_1+j_2 + j_3= k} \bar g_1(x_{j_1}) \bar g_2(x_{j_2}) \bar g_3(x_{j_3}) =
\Big((\bar g_1 \circledast \bar g_2) \circledast \bar g_3 \Big)(x_k),
\]
This shows that $\bar f^{\circledast n}$ is a well-defined operation for any $n\ge 2$, 
since it does not matter which order the convolutions are taken in. So for any $n>2$
and mesh point $x_k$,
\begin{equation}\label{eq:associativeDiscreteConvProp}
\begin{split}
\bar f^{\circledast n}(x_k) &= (\bar f^{\circledast (n-1)} \circledast \bar f )(x_k)
= (\bar f^{\circledast (n-2)} \circledast \bar f^{\circledast 2})(x_k)\\ &
= \cdots = (\bar f \circledast \bar f^{\circledast (n-1)})(x_k)
= h^{n-1} \sum_{j_1 + \cdots +j_{n} = k} \bar f(x_{j_1}) \cdots \bar f(x_{j_{n}}) \,.
\end{split}
\end{equation}

\subsection{Efficient computation of \texorpdfstring{\(\bar f^{\circledast n}\)}{bar f conv}}
\label{sec:effComp}
In this section we describe the  sequence of discrete convolutions used to compute $\bar f^{\circledast n}$ efficiently. This matters for the performance of the method when implemented on a computer. 

Let $m\in \mathbb{N}$ denote the largest integer such that $m \le \log_2(n)$ and 
for $\ell =2,\ldots,m$, we minimize the number of convolutions through computing
\begin{equation}\label{eq:discLinConv}
\bar f^{\circledast 2^{\ell}}(x_k) = h \sum_{j=0}^{k} \bar f^{\circledast2^{\ell-1}}(x_j) \bar f^{\circledast2^{\ell-1}}(x_{k} - x_j) \qquad k = 0,\ldots, N. 
\end{equation}
If $2^m=n$, then the above computation represents discrete approximation $\bar f^{\circledast n}$ of the density of
$\sum_{j=1}^n X_j$, otherwise we obtain the approximation by 
\[
\bar f^{\circledast n}(x_k) = h \sum_{j=0}^{k} \bar f^{\circledast2^{m}}(x_j) \bar f^{\circledast(n-2^{m})}(x_{k} - x_j) \qquad k = 0,\ldots, N,
\]
where $\bar f^{\circledast(n-2^{m})}$ is computed in at most $m-1$ steps by a similar iterative application of the Trapezoidal rule.

Lastly, the probability of failure is also approximated by a suitable
closed Newton--Cotes formula:
\begin{equation}\label{eq:newtonCotesAlpha}
\bar \alpha_N := \sum_{j=0}^{N} w_j \bar f^{\circledast n}(x_j),
\end{equation}
where the weights sum to the the length of the interval $[0,\gamma]$, meaning  
$\sum_{j=0}^N w_j = \gamma$, and we
restrict ourselves to the class of Newton--Cotes formulas with
non-negative weights, which means formulas with degree $d \le 8$,
cf.~\cite[Chapter 7]{suli2003introduction}.  See
Theorem~\ref{thm:errorEstimate} for further details on how to choose
the Newton--Cotes rule.
Since it holds that $\bar f^{\circledast j}(0)=0$ for any $j\ge1$, all of the
numerical integration above can be viewed as applications of the Trapezoidal rule.

\begin{table}
    \centering
    \begin{tabular}{c|l}
        \hline
        Fading Type & PDF \\
        \hline
        Rayleigh & $\frac{2x}{\Omega} \exp \left (-\frac{x^2}{\Omega} \right )$\\
        Nakagami-m & $\frac{2m^m}{\Omega^m\Gamma(m)} x^{2m-1}\exp \left (-\frac{m}{\Omega} x^2\right )$\\
        Rice & $ \frac{2(K+1)x}{\Omega}\exp \left (-K-\frac{K+1}{\Omega}x^2 \right )I_0 \left (2\sqrt{\frac{K(K+1)}{\Omega}} x\right )$\\
        Weibull & $ k \left (\frac{\beta}{\Omega} \right )^{k}x^{k-1} \exp \left (-\left (\frac{x\beta}{\Omega}\right)^k \right ) $ , \text{  } $\beta=\Gamma\left ( 1+\frac{1}{k}\right )$\\
        Log-Normal & $\frac{1}{x\sigma\sqrt{2\pi}}\exp \left (-\frac{(\log(x)-\mu)^2}{2\sigma^2} \right )$\\
        generalized Gamma & $\frac{p(\frac{\beta}{\Omega})^{d}}{\Gamma(\frac{d}{p})} x^{d-1} \exp \left (-(\frac{\beta}{\Omega}x)^{p} \right )$, \text{  }$\beta=\frac{\Gamma \left (\frac{d+1}{p} \right )}{\Gamma \left (\frac{d}{p} \right )}$\\
        $\kappa-\mu$ &  $\frac{2\mu(K+1)^{\frac{1+\mu}{2}}x^{\mu}}{\Omega^{\frac{1+\mu}{2}}\kappa^{\frac{\mu-1}{2}}}\exp \left (-\mu K-\frac{K+1}{\Omega}\mu x^2 \right )I_{\mu-1} \left (2\mu\sqrt{\frac{K(K+1)}{\Omega}} x\right )$\\
        \hline
    \end{tabular}
    \caption{Some common PDF satisfying $f(0)=0$. Functions $\Gamma(\cdot)$ and $I_{\xi}(\cdot)$ are respectively the Gamma function and the modified Bessel function of the first kind and order $\xi$ \cite{gradshteyn2007}}
    \label{table1}
\end{table}

\begin{remark}\hspace{1mm} The assumption that $f(0)=0$ is not restrictive as it is satisfied by most of the common fading channel envelopes. A non exhaustive list of common fading channel having the previous assumption satisfied is in Table I. 
\end{remark}

\begin{remark}
For the simpler setting when $f$ is a probability mass function instead of a
probability density function, a similar approach to approximating the density 
of sums $Y_1 + \cdots + Y_n$ where $Y_k \stackrel{iid}{\sim} f$ through FFT-based 
convolution has been studied in~\cite{wilson2017accurate}.
\end{remark}

\subsection{Implementations of discrete linear convolution}
There are two standard approaches to implement the above discrete linear convolution~\eqref{eq:discLinConv}
in most programming languages, with different strengths and weaknesses: 

\begin{enumerate}

\item \textbf{Direct convolution:} For each $k=0,1,\ldots,N$ compute the RHS sliding sum of~\eqref{eq:discLinConv}. In Matlab, this can be achieved by calling the $\textbf{conv}()$ function as follows:
  \begin{verbatim}
    fBar_2l = conv(g,g);
    fBar_2l   = fBar_2l(1:N+1)*h; 
  \end{verbatim}
  for an input vector $g:=\bar f^{2^{\ell-1}\circledast} \in \mathbb{R}^{N+1}$. The computational cost of this function call, measured in number of floating point operations is $\mathcal{O}(N^2)$, which is quite high. But Theorem~\ref{thm:errorWithRounding} and our numerical experiments in Section~\ref{sec:convVsFFT} show that direct convolution is accurate even for very rare events and it does not appear sensitive to round-off errors.

 \item \textbf{FFT-based convolution:} The second approach is to append/pad $N$ zeros to the vector $f^{\circledast2^{\ell-1}} \in \mathbb{R}^{N+1}$, and use the fast Fourier transform (FFT) to compute 
  the linear convolution as follows: 
  \begin{equation}\label{eq:fftMethod}
    \begin{split}
      \bar f^{\circledast2^{\ell-1}} &= [\bar f^{\circledast2^{\ell-1}}(x_0), \ldots, \bar f^{\circledast2^{\ell}}(x_{N}), \underbrace{0,\ldots, 0}_{N}] \\
      \bar f^{\circledast2^{\ell}} &=  \textbf{IFFT}(\textbf{FFT}(\bar f^{\circledast2^{\ell-1}}).\text{ \textasciicircum}^2 ) \times h  \\
      \bar f^{\circledast2^{\ell}} &= [\bar f^{\circledast2^{\ell}}(x_0), \ldots, \bar f^{\circledast2^{\ell}}(x_N)] \, .
    \end{split}
  \end{equation}
  which in Matlab takes the form 
  \begin{verbatim}
    g       = [g zeros(1,N)];
    fBar_2l = ifft(fft(g).^2)*h;
    fBar_2l = fBar_2l(1:N+1);
   \end{verbatim}
  An advantage with this approach is that the computational cost of the three assignments~\eqref{eq:fftMethod} is $\mathcal{O}(N \log(N))$, which for large $N$ will be far lower than the cost of direct convolution. A downside is that FFT-based convolution 
  can be very sensitive to rounding errors when the floating point precision is 
  low, which occurs when the machine epsilon is greater or of the similar magnitude as $\alpha$, cf.~Section~\ref{sec:roundingErrors}. This is illustrated in an numerical example in Section \ref{sec:convVsFFT}.

  \end{enumerate}

\begin{remark}
The method straightforwardly extends to settings where 
$X_1,\ldots, X_n$ are independent but not identically distributed, cf.~\eqref{eq:discLinConv2}.
\end{remark}

\section{Theoretical results}\label{sec:theory}

In this section we first prove that $\bar f^{\circledast n}(x_k) \to f^{*n}(x_k)$ and 
that $\bar{\alpha}_N \to \alpha$ as $N \to \infty$ and obtain convergence rates for 
these results in Lemma~\ref{lem:fNumConvErr} and Theorem~\ref{thm:errorEstimate}, respectively. 
Thereafter, we bound the relative approximation error of $\fl{\bar \alpha_N} \approx \bar \alpha_N$ for direct convolution in terms of the floating-point precision machine epsilon in Lemma~\ref{lem:roundingErrorDirectConv}. This leads to the upper bound for the computable approximation $\fl{\bar \alpha_N} \approx \alpha$ that is 
described in Theorem~\ref{thm:errorWithRounding}. A similar results for FFT-based convolution is given in \ref{thm:roundingErrorFFT}. Finally, we bound the computational cost of our method in Theorem~\ref{thm:costEstimate} and combine the results on error and cost to compare the efficiency of direct convolution and FFT-based convolution in~\eqref{eq:costVsError}.

Let $C_{0}^{2p}([0,\gamma])$ denote the set of $2p$ times continuously differentiable
functions on $[0,\gamma]$ for which $f^{(k)}(0) = 0$ for all $k \in \{0,1,\ldots,2p-1\}$, and let $C_{0,\per}^{2p}([0,\gamma]) \subset C_{0}^{2p}([0,\gamma])$ denote the subset of such functions that also are periodic on $[0,\gamma]$ up to the $(2p-1)$-th derivative, meaning that $f^{(k)}(0) = f^{(k)}(\gamma)$ for all $k \in \{0,1,\ldots,2p-1\}$.

The first Lemma provides a convergence rate for $\bar f^{\circledast m}(x_k) \to f^{*m}(x_k)$ as $N\to \infty$. It does not take rounding errors into account.  
\begin{lemma}\label{lem:fNumConvErr}
Let $f \in C_{0}^{2p}([0,\gamma])$ for some integer $p \ge 1$ and let $f^{\circledast n}(x_k)$
for $n\ge 2$ be defined by~\eqref{eq:associativeDiscreteConvProp}.  Then, there exists a constant $C_1>0$ that depends on $p$ such that 
\begin{equation}\label{eq:convRateFNumConv}
|\bar f^{\circledast n}(x_k) - f^{*n}(x_k)| \le (1+x_k)^{n -2} C_1 C_2(n,x_k) N^{-2p} \text{ for all } k =0,1, \ldots, N, 
\end{equation}
where 
\begin{multline*}
C_2(n,x_k) :=\\ 
\max_{m=2,\ldots,n} \max_{x_{j_1} + \cdots + x_{j_{m-1}} = x_k} \; \max_{ 0\le y \le x_{j_1}}  \left| \frac{d^{2p}}{dy^{2p}} f^{*(n+1-m)}(y) f(x_{j_1}-y)  \right| \prod_{\ell =2}^{m-1} f(x_{j_\ell})\,
\end{multline*}
with the conventions that $\prod_{\ell =2}^{m-1} f(x_{j_\ell}) \equiv 1$  when $m =2$
and $f^{*1} = f$.
\end{lemma}

\begin{proof}
Assume that~\eqref{eq:convRateFNumConv} holds for all $ 2\le n \le \bar n$ for some $\bar n \ge 2$. Recalling from~\eqref{eq:associativeDiscreteConvProp} that 
$\bar f^{\circledast (\bar n+1)}(x_k) = (\bar f^{\circledast \bar n} \circledast \bar f)(x_k)$
and also that $\bar f(x_j) = f(x_j)$ for all $j =0,1,\ldots, N$,
we obtain
\begin{align*}
|&\bar f^{\circledast (\bar n+1)}(x_k) - f^{*( \bar n+1)}(x_k)| =\\
&\hspace{31.75mm} h \left|\sum_{j=0}^k  \bar f^{\circledast \bar n}(x_j) \bar f(x_{k-j}) - f^{\ast \bar n}(x_j) f(x_{k-j})\right|\\
&\hspace{30mm} + \left|h \sum_{j=0}^k f^{\ast \bar n}(x_j)f(x_{k-j}) - \int_0^{x_k} f^{\ast \bar n}(y) f(x_k-y) dy \right| =: I + II\, . 
\end{align*}

For the first term, 
\[
\begin{split}
I &= h \left|\sum_{j=0}^k  \bar f^{\circledast \bar n}(x_j) \bar f(x_{k-j}) - f^{\ast \bar n}(x_j) f(x_{k-j})\right|\\
& \le  h \sum_{j=0}^k  |\bar f^{\circledast \bar n}(x_j) -  f^{\ast \bar n}(x_{j})| f(x_{k -j})\\
& \le  h C_1 (1+x_k)^{\bar n -2}\sum_{j=0}^{k} C_2(\bar n,x_j) f(x_{k-j}) \\
& \le (k-1)h \times (1+x_k)^{\bar n -2} C_1 C_2(\bar n +1 , x_k) N^{-2p} \\
&\le x_k (1+x_k)^{\bar n -2} C_1 C_2(\bar n+1, x_k) N^{-2p},
\end{split}
\]
where the penultimate inequality follows from the change of subindex in 
$x_{k-j} = x_{j_m}$ and
\[
\begin{split}
&f(x_{k-j}) C_2(\bar n,x_j) \\
&= f(x_{k-j})  \max_{m=2,\ldots, \bar n} \max_{x_{j_1} + \cdots + x_{j_{m-1}} = x_j} \; \max_{ 0\le y \le x_{j_1}}  \left| \frac{d^{2p}}{dy^{2p}} f^{*(\bar n +1-m)}(y) f(x_{j_1}-y)  \right| \prod_{\ell =2}^{m-1} f(x_{j_\ell})\\
&\le \max_{m=2,\ldots,\bar n} \max_{x_{j_1} + \cdots + x_{j_{m}} = x_k} \; \max_{ 0\le y \le x_{j_1}}  \left| \frac{d^{2p}}{dy^{2p}} f^{*(\bar n+1-m)}(y) f(x_{j_1}-y)  \right| \prod_{\ell =2}^{m} f(x_{j_\ell})\\
&= \max_{m=3,\ldots,\bar n+1} \max_{x_{j_1} + \cdots + x_{j_{m-1}} = x_k} \; \max_{ 0\le y \le x_{j_1}}  \left| \frac{d^{2p}}{dy^{2p}} f^{*(\bar n+2-m)}(y) f(x_{j_1}-y)  \right| \prod_{\ell =2}^{m-1} f(x_{j_\ell})\\
&\le C_2(\bar n +1, x_k)\, .
\end{split}
\]\\
The second term is the quadrature error of the composite Trapezoidal rule applied 
to the integrand \mbox{$g(y) = f^{*\bar n}(y) f(x_k-y)$.} Thanks to $g \in C^{2p}_{0,\per}([0,x_k])$,
we obtain that 
\[
\begin{split}
II \le  C_1 \max_{ 0\le y \le x_{k}}  \left| \frac{d^{2p}}{dy^{2p}} f^{*\bar n}(y) f(x_k-y) \right| N^{-2p} \le C_1C_2(\bar n+1, x_k) N^{-2p},
\end{split}
\]
for the constant $C_1>0$ introduced above. This yields, 
\[
|\bar f^{\circledast (\bar n+1)}(x_k) - f^{*( \bar n+1)}(x_k)| 
\le (1+x_k)^{\bar n-1} C_1C_2(\bar n+1,x_k)\, . 
\]\\
We next verify that~\eqref{eq:convRateFNumConv} holds for $n=2$. Since $\bar f^{\circledast 2}(x_k)$
is the composite Trapezoidal rule approximation of $f^{*2}(x_k)$, as is apparent from 
\[
|\bar f^{\circledast 2}(x_k) - f^{*2}(x_k)| = \left|h \sum_{j=0}^k f(x_j)f(x_{k-j}) - \int_0^{x_k} f(y) f(x-y) dy \right|,
\]
and $f \in C^{2p}_0([0, \gamma])$ implies that $g(y):= f(y) f(x_k-y)$ belongs to $C^{2p}_{0,\per}([0, x_k])$, it follows from~\cite[Chapter 7.6]{suli2003introduction} that 
\[
|\bar f^{\circledast 2}(x_k) - f^{*2}(x_k)| \le C_1 \underbrace{\max_{ 0\le y \le x_{k}}  \left| \frac{d^{2p}}{dy^{2p}} f(y) f(x_k-y) \right|}_{= C_2(2,x_k)} N^{-2p},
\]
for the constant $C_1>0$ introduced above. The proof follows by induction. 
\end{proof}

The next theorem proves a convergence rate for $\alpha_N \to \alpha$ as $N\to \infty$ in the setting of no rounding errors. 

\begin{theorem}\label{thm:errorEstimate}
  Let $f \in C_{0}^{2p}([0,\gamma])$ for some integer $p \ge 1$, and let $r \le 2p$ be the order of convergence for the Newton--Cotes rule with non-negative weights that is used to compute $\bar \alpha$ in~\eqref{eq:newtonCotesAlpha}. 
  Then there exist a constant $C_3>0$ that depends on $r$ such that
  \begin{equation}\label{eq:convRate}
    |\bar \alpha_N - \alpha | \le \gamma  C_1 \overline{C}_2(n,\gamma) N^{-2p} + C_3 \max_{x \in [0,\gamma]} \left|\frac{d^r}{dx^r} f^{*n}(x)\right| N^{-r},  
  \end{equation}
  where $\overline{C}_2(n,\gamma) := \max_{k=0,\ldots,N} (1+x_k)^{n-2} C_2(n,x_k)$,
  and the constant $C_1$ and the mapping $C_2(n,x_k)$ are defined in Lemma~\ref{lem:fNumConvErr}. 
  
\end{theorem}

\begin{proof}
  For a closed Newton--Cotes formula with convergence rate $r\le 2p$, it follows
  by~\cite[Chapter 7.1.1]{isaacson2012analysis} that
  \[
    \left|\alpha - \sum_{j=0}^N w_j f^{*n}(x_j) \right| \le C_3 N^{-r}.
  \]
  The triangle inequality and the non-negativity of the weights $w_j$
  yield the final bound
  \[
    \begin{split}
      |\alpha - \bar \alpha_N|&  = \left|\alpha - \sum_{j=0}^N w_j \bar f^{\circledast n}(x_j) \right| \\
      & \le \left|\alpha - \sum_{j=0}^N w_j f^{*n}(x_j) \right|  +
        \sum_{j=0}^N w_j |f^{*n}(x_j) - \bar f^{\circledast n}(x_j)| \\
      & \le  C_3 \max_{x \in [0,\gamma]} \left|\frac{d^r}{dx^r} f^{*n}(x)\right| N^{-r} + \sum_{j=0}^N w_j (1+x_k)^{n-2} C_1 C_2(n,x_j) N^{-2p} \\
      & \le C_3 \max_{x \in [0,\gamma]} \left|\frac{d^r}{dx^r} f^{*n}(x)\right| N^{-r} + \gamma  C_1 \overline{C}_2(n,\gamma) N^{-2p} \\
    \end{split}
  \]
\end{proof}

\begin{remark} \label{rem:f_not_zero_at_zero}
If \(f \notin C^{2p}_0[0, \gamma]\), but we have \(f \in C^2[0, \gamma]\) (which is the case if e.g. $f(0) \neq 0$ or $f'(0) \neq 0$), then the slightly altered direct convolution (compare to~\eqref{eq:discLinConv}) 
\begin{multline}\label{eq:discLinConv2}
\bar f^{\circledast 2^{\ell}}(x_k) = 
h \sum_{j=1}^{k-1} \bar f^{\circledast2^{\ell-1}}(x_j) \bar f^{\circledast2^{\ell-1}}(x_{k} - x_j) \\
+ h \frac{\bar f^{\circledast2^{\ell-1}}(x_0) \bar f^{\circledast2^{\ell-1}}(x_{k}) +
\bar f^{\circledast2^{\ell-1}}(x_k) \bar f^{\circledast2^{\ell-1}}(x_0)}{2}
\end{multline}
lead to the error bounds as in Lemma~\ref{lem:fNumConvErr} and Theorem~\ref{thm:errorEstimate} with $p=1$. 
\end{remark}

\subsection{Rounding errors} \label{sec:roundingErrors}
In practice, approximations of $\alpha$ are computed using floating-point arithmetic where a float is represented by $x = s \times b \times 2^e$ with sign $s$ (1-bit), the significand $b \in [1,2)$ and exponent $e$. The standard IEEE 754 64-bit floats, for example, has $p=53$-bit significand precision (52-bits stored) and 11-bit exponent. For estimating rounding errors in relative error, the machine epsilon $\varepsilon = 2^{-p}$, which is equal to half the distance between the number $1$ and the closest floating point number to $1$, is important. For an $x \in \mathbb{R}$, let $\fl{x}$ denote the closest number to $x$ among the floating point numbers. Then it holds that 
$|x - \fl{x}| \le (1 + \varepsilon)|x|$.  

More generally, we let 
$\fl{\bar{\alpha}_N}$ denote the value of $\bar{\alpha}_N$ that 
is obtained when \textbf{all underlying arithmetic operations} are computed with
the given floating point precision, and thus possibly all being subject to rounding errors, 
and similarly also for $\fl{\bar f^{\circledast n}(x_k)}$.
Observe that this notation is recursive, it assumes that a quantity is computed in a uniquely specified manner (otherwise it would not be clear how to estimate rounding errors), and it is extremely compact, as 
is illustrated when applying it to the formula~\eqref{eq:discLinConv}:
\[
\fl{\bar f^{\circledast 2^\ell}(x_k)} = \fl{\fl{h}  \fl{\sum_{j=0}^{k} \fl{\fl{\bar f^{\circledast2^{\ell-1}}(x_j)} \fl{\bar f^{\circledast 2^{\ell -1}}(x_{k} - x_j)} }}}.
\]

\begin{lemma}[Rounding error direct convolution]\label{lem:roundingErrorDirectConv}
Let $\bar \alpha_N$ be computed by direct convolution, $\fl{\bar \alpha_N}$ be computed with floating point arithmetic with machine epsilon $\varepsilon>0$ and let \(n\) be the number of i.i.d RVs in the underlying sum. Furthermore, set 
\(\bar{m} = \lceil \log_2(n) \rceil\).
Assume that for each \(x\) in the codomain of \(f\) we have \(\abs{\fl{x}-x} \leq xc\varepsilon\) for some \(c \in (0,1]\). Moreover, assume that \(N \geq 2^{10}\), \(\varepsilon \leq 2^{-53}\) and that \(2^{2\bar{m}+2}N \varepsilon<1/10\). Then it holds that 
\[\lvert \fl{\bar \alpha_N} - \bar \alpha_N \rvert \le 4 \bar\alpha_N n N \varepsilon.\]
\end{lemma}

\begin{proof}
We begin the proof by looking at some results from \cite{keich2005sfft} showing how rounding errors propagate when adding and multiplying already estimated values. Assume that we have a set of non-negative real numbers \(A = \{a_1, a_2, \hdots, a_N\}\) that are estimated by use of floating point arithmetic with an arbitrary number of operations used for calculating the approximations. We denote the estimated values \(\tilde A = \{\tilde a_1, \tilde a_2, \hdots, \tilde a_N\}\) and have that the absolute accumulated error of our estimates \(\tilde a_i\) can be bounded by some constant \(c_a > 0\), giving \(\abs{\tilde a_i - a_i} \leq \abs{a_i}c_a\varepsilon\). Moreover, let \(\fl{h}\) be our floating point estimation of the step length \(h \in \mathbb{R}_{>0}\) with \(\abs{\fl{h}-h} \leq h \varepsilon\). 
Based on the proof of \cite[Lemma 3]{keich2005sfft}, it is straightforward to check that
\begin{equation}\label{eq:multiplicationError}
    \abs{\fl{\fl{h}\tilde a_i \tilde a_j}-ha_ia_j} \leq ha_ia_j \left [2 + 2c_a + (1 + 4c_a+c_a^2)\varepsilon + (2c_a + 2c_a^2)\varepsilon^2 + c_a^2\varepsilon^3 \right ]\varepsilon,
\end{equation}
for \(i,j \in \{1, 2, \hdots, N\}\). Then, by letting \(S_k = \sum_{j=0}^k ha_ja_{k-j}, k \in \{1, 2, \hdots, N\}\) we have from \cite[Lemma 2]{keich2005sfft} that
\begin{multline}
    \label{eq:sumError}
    \abs{\fl{\sum_{j=0}^k \fl{\fl{h}\tilde a_i \tilde a_{k-j}}} - S_k} \leq\\
    S_k \left [k + 2 + 2c_a + (1 + 4c_a+c_a^2)\varepsilon + (2c_a + 2c_a^2)\varepsilon^2 + c_a^2\varepsilon^3 \right ]\varepsilon(1 + k\varepsilon),
\end{multline}
when \eqref{eq:multiplicationError} holds for each term in the sum and \((N + c_{ha_ja_{k-j}})\varepsilon < 1\), where 
\[c_{ha_ja_{k-j}} = 2 + 2c_a + (1 + 4c_a+c_a^2)\varepsilon + (2c_a + 2c_a^2)\varepsilon^2 + c_a^2\varepsilon^3.\]
Note that \ref{eq:sumError} only holds as long as all terms of \(S_k\) have the same sign, which in our case is non-negative.
For simplicity, we will bound the rounding error of all \(S_k\) by inserting \(N\) instead of \(k\) in \ref{eq:sumError}, yielding
\begin{multline*}
\abs{\fl{\sum_{j=0}^k \fl{\fl{h}\tilde a_i \tilde a_{k-j}}} - S_k} \leq \\ S_k \left [N + 2 + 2c_a + (1 + 4c_a+c_a^2)\varepsilon + (2c_a + 2c_a^2)\varepsilon^2 + c_a^2\varepsilon^3 \right ]\varepsilon(1 + N\varepsilon)
\end{multline*}
\(\text{ for all } k \in \{0, 1, \hdots N\}.\)

We now move on to prove the following statement by induction on \(l\):
\begin{equation} \label{eq:induction}
    \abs{\fl{\bar f^{\circledast 2^l}(x_k)}-\bar f^{\circledast 2^l}(x_k)} \leq \bar f^{\circledast 2^l}(x_k)(2^{l+1}-2)N\varepsilon.
\end{equation}

For \(l=1\) we first observe that from \eqref{eq:multiplicationError} and our assumptions we have 
\begin{align*}
    \abs{\fl{h\bar f(x_j) \bar f(x_{k-j})}-h\bar f(x_j) \bar f(x_{k-j})} \leq& h\bar f(x_j) \bar f(x_{k-j})[2 + 2c + (1 + 4c+c^2)\varepsilon \\
    &\hspace{25mm}+ (2c + 2c^2)\varepsilon^2 + c^2\varepsilon^3]\varepsilon\\
    \leq& h\bar f(x_j) \bar f(x_{k-j})5\varepsilon.
\end{align*}
Then, from \eqref{eq:sumError} we have that
\begin{align*}
    \abs{\fl{\bar f^{\circledast 2}(x_k)}-\bar f^{\circledast 2}(x_k)} \leq& \bar f^{\circledast 2}(x_k) (N+5)(1+N\varepsilon)\varepsilon \\ \leq& \bar f^{\circledast 2}(x_k)\left(N+5+N^2\varepsilon + 5N\varepsilon \right)\varepsilon \leq \bar f^{\circledast 2}(x_k) 2N\varepsilon,
\end{align*}
as we needed to show. 

Assume now that \eqref{eq:induction} holds for some \(q \in \mathbb{N} \setminus \{0\}\), that is 
\[\abs{\fl{\bar f^{\circledast 2^q}(x_k)}-\bar f^{\circledast 2^q}(x_k)} \leq \bar f^{\circledast 2^q}(x_k)(2^{q+1}-2)N \varepsilon\]
Furthermore, assume that given some value for \(N\) our \(q\) satisfies \(2^{2q+2}N \varepsilon<1/10\). Then we have from \eqref{eq:sumError} that
\begin{align*}
    \abs{\fl{\bar f^{\circledast 2^{q+1}}(x_k)}-\bar f^{\circledast 2^{q+1}}(x_k)} \leq& \bar f^{\circledast 2^{q+1}}(x_k) \bigg [N + 2 + 2\left(2^{q+1}-2\right)N \\
    &\hspace{11.5mm}+ \left (1 + 4\left[2^{q+1}-2\right]N+\left[2^{q+1}-2\right]^2N^2 \right )\varepsilon \\
    &\hspace{11.5mm}+ \left (2\left[2^{q+1}-2\right]N + 2\left[2^{q+1}-2\right]^2N^2 \right )\varepsilon^2 \\ 
    &\hspace{11.5mm}+ \left(2^{q+1}-2\right)^2N^2\varepsilon^3 \bigg ] \varepsilon(1 + N\varepsilon) \\
    \leq& \bar f^{\circledast 2^{q+1}}(x_k) \bigg [2 + \left(2^{q+2}-3\right)N + \varepsilon + \frac{1}{10} + \frac{N}{10} \\
    &\hspace{20.5mm}+ \frac{\varepsilon}{10} + \frac{2N\varepsilon}{10} + \frac{N\varepsilon^2}{10} \bigg ] \varepsilon(1 + N\varepsilon) \\
    \leq& \bar f^{\circledast 2^{q+1}}(x_k) \left [\left(2^{q+2}-\frac{5}{2} \right)N + \frac{N}{10} \right ]\varepsilon \\
    \leq& \bar f^{\circledast 2^{q+1}}(x_k) \left (2^{q+2}-2\right)N\varepsilon.
\end{align*}

Proceeding with the last step we need to calculate the accumulated rounding error for \(\bar \alpha_N = \sum_{j = 0}^N w_j \bar f^{\circledast n}(x_j)\). 
We then set \(m = \lfloor \log_2(n) \rfloor\). Next we need to consider two separate cases, one where \(m = \log_2(n)\), and the alternative case \(m < \log_2(n)\). In the former case we have
\[\bar \alpha_N = \sum_{j = 0}^N w_j \bar f^{\circledast 2^{\log_2(n)}}(x_j) = \sum_{j = 0}^N w_j \bar f^{\circledast 2^{m}}(x_j).\] 
In the latter case we calculate \(\bar f^{\circledast 2^{m}}(x_j)\) and \(\bar f^{\circledast (n-2^{m})}(x_j)\), which can be done in at most \(m-1\) steps. 
We then have that \(\bar f^{\circledast n}(x_j) = \bar f^{\circledast 2^m}(x_j) \circledast \bar f^{\circledast n-2^m}(x_j)\), which would have an error bounded by \(\bar f^{\circledast 2^{m+1}}\). 
Therefore, by setting \(\bar m = \lceil \log_2(n) \rceil\), we have that the error of \(\bar f^{\circledast n}(x_j)\) is bounded by the error of \(\bar f^{\circledast 2^{\bar m}}(x_j)\).

Moving on we have from \cite[Lemma 3]{keich2005sfft}, the bound \eqref{eq:induction} and from our assumptions that
\begin{align*}
    \abs{\fl{w_j \bar f^{\circledast n}(x_j)}-w_j \bar f^{\circledast n}(x_j)} \leq& w_j \bar  f^{\circledast 2^{\bar m}}(x_j) [2 + \left (2^{\bar m+1}-2 \right)N \\ 
    &\hspace{21.25mm}+ \left(1 + \left (2^{\bar m+2}-4 \right)N\right)\varepsilon \\
    &\hspace{21.25mm}+ \left (2^{\bar m+1}-2\right )N\varepsilon^2 ]\varepsilon \\
    \leq& w_j \bar  f^{\circledast 2^{\bar m}}(x_j) \left [2 + \left (2^{\bar m+1}-2 \right)N + \varepsilon + \frac{1}{10} + \frac{\varepsilon}{10}\right ]\varepsilon \\
    \leq& w_j \bar  f^{\circledast 2^{\bar m}}(x_j) \left [\left (2^{\bar m+1}-2 \right)N + 3\right ]\varepsilon.
\end{align*}

Then, by \cite[Lemma 2]{keich2005sfft} we get the following bound for the rounding error of \(\bar \alpha_N\):
\begin{align*}
    \abs{\fl{\bar \alpha_N} - \bar \alpha_N} \leq& \bar \alpha_N \left [\left (2^{\bar m+1}-2 \right)N + N + 3\right ](1+N\varepsilon) \varepsilon \\
    \leq& \bar \alpha_N \left [\left (2^{\bar m+1}-1 \right)N + 3 + \frac{N}{10} + 3N\varepsilon\right]\varepsilon \\
    \leq& \bar \alpha_N 2^{\bar m+1}N\varepsilon \\
    \leq& 4\bar\alpha_NnN\varepsilon 
\end{align*}
which is what we set out to prove.
\end{proof}

This leads to our main convergence result.
\begin{theorem}[Approximation error direct convolution]\label{thm:errorWithRounding}
Let the assumptions in Theorem~\ref{thm:errorEstimate} and Lemma~\ref{lem:roundingErrorDirectConv} hold. 
Then it holds that 
\begin{multline*}
    |\fl{\bar \alpha_N} - \alpha| \le 
(1 + 4 n N \epsilon) \left(\gamma C_1 \overline{C}_2(n,\gamma) N^{-2p} + C_3 \max_{x \in [0,\gamma]} \left|\frac{d^r}{dx^r} f^{*n}(x)\right| N^{-r}\right) \\+ 4 \alpha n N\varepsilon,
\end{multline*}
for all integers $n$ and $N$ such that $N \le \varepsilon/C_5$
and \( (4n)^2 N \varepsilon<1/10\). 
\end{theorem}
The result follows from 
Lemma~\ref{lem:roundingErrorDirectConv} and Theorem~\ref{thm:errorEstimate} 
and using the triangle inequality.

We continue with a lemma needed for the proceeding result. The lemma is a version of \cite[Lemma 5]{keich2005sfft} and the proof of our lemma is based on the one given in the cited paper. Note first that the discrete version \(\bar f\) of \(f\) can be associated with a vector \(q \in \mathbb{R}_{\geq 0}^N\) by letting \(q_i = \bar f(x_i)\). Then we let
\[\norm{q}_1 \coloneqq \sum_{i=1}^N \abs{q_i} \text{ and } \norm{q}_\infty \coloneqq \max_{1 \leq i \leq N} \abs{q_i}.\]
We also need to define the DFT and IDFT operators, denoted \(D_{N}\) and \(D^{-1}_N\) respectively. Let 
\[D_{N, k,j} = e^{\frac{ikj2\pi}{N}}, \text{ and } D^{-1}_{N, k,j} = \frac{1}{N}e^{\frac{-ikj2\pi}{N}}.\]
We will denote the operators by \(D\) and \(D^{-1}\) when the dimension is clear from the context. 
We are then ready to proceed with the lemma. 
\begin{lemma} 
    \label{lem:singelFFTConvBound}
    Let \(q \in \mathbb{R}_{\geq 0}^N\) where the numerical approximation \(\fl{q} \in \mathbb{R}^N\) satisfies \(\abs{q_i-\fl{q_i}} \leq q_i c_\delta \varepsilon\) for some \(c_\delta \in (0,1]\) and let \(k = \lceil \log_2(N) \rceil + 1\). Assume that \(\fl{q^{\circledast 2}}\) is computed by the method described in \eqref{eq:fftMethod}, that \(13k\varepsilon \leq 1\) as well as \(\abs{\fl{h}-h} = 0\), i.e. that we are able to accurately represent the constant \(h = \frac{\gamma}{N}\) numerically. Furthermore, we also assume that \(\fl{D}\fl{q}\) and its square can be calculated exactly with floating point arithmetic, that is \(\fl{D}\fl{q} = \fl{\fl{D}\fl{q}}\) and \((\fl{D}\fl{q})^2 = \fl{(\fl{D}\fl{q})^2}\). Then
    \[\norm{\fl{q^{\circledast 2}} - q^{\circledast 2}}_\infty \leq 2h(c_\delta+9k)\varepsilon\norm{q}_1^2+ch\varepsilon^2,\]
    where \(c > 0\) is a constant depending on \(k\) and \(\norm{q}^2_1\) capturing higher-order terms of \(\varepsilon\).
\end{lemma}
\begin{proof}
    Let \(q \in \mathbb{R}_{\geq 0}^{2^k}\) be a zero-padded version of our original \(q \in \mathbb{R}_{\geq 0}^N\). Note that we for simplicity collect all higher-order terms of \(\varepsilon\) in constants \(c_j\) throughout this proof. We then have that
\begin{align*}
    \norm{\fl{D}\fl{q} - Dq}_\infty \leq& \norm{D(\fl{q} - q)}_\infty + \norm{(D - \fl{D})\fl{q}}_\infty \\
    \leq& \norm{\fl{q} - q}_1 + 6k\varepsilon \norm{\fl{q}}_1 \\
    \leq& c_\delta\varepsilon\norm{q}_1 + 6k\varepsilon(1+c_\delta\varepsilon)\norm{q}_1 \\
    \leq& (c_\delta + 6k)\varepsilon \norm{q}_1 + c_1 \varepsilon^2
\end{align*}
where, for the first inequality, we have used the triangle inequality and the transition from the first to the second line holds due to the fact that \(\norm{Dx}_\infty \leq \norm{x}_1,\, x\in \mathbb{R}^{2^k}\) and \cite[Lemma 4]{keich2005sfft} together with our assumption on \(13k\varepsilon\). The jump from the second to the third line comes from our assumptions on the differences \(\abs{q_i-\fl{q_i}} \leq q_i c_\delta \varepsilon\) as this implies
\[\norm{\fl{q} - q}_1 = \sum_{i=0}^{2^k}\abs{\fl{q(x_i)} - q(x_i)} \leq \sum_{i=0}^{2^k}q(x_i)c_\delta\varepsilon \leq c_\delta\varepsilon\norm{q}_1,\]
and further
\begin{align*}
\norm{\fl{q}}_1 =& \sum_{i=0}^{2^k}\abs{\fl{q(x_i)}} \leq \sum_{i=0}^{2^k} \left(\abs{\fl{q(x_i)}-q(x_i)}+q(x_i)\right) \\
\leq& \norm{\fl{q} - q}_1 + \norm{q}_1 = (1+c_\delta\varepsilon)\norm{q}_1.
\end{align*}
The final transition follows by choosing an appropriate constant \(c_1\). Let now \(r(x)=\left[(Dq)(x)\right]^2\) and similarly \(\fl{r(x)}=\left[(\fl{D} \fl{q})(x)\right]^2\), where we have used our assumption stating that the multiplication of \(\fl{D}\), \(\fl{q}\) and the square of their product can be represented exactly in floating point arithmetic.
Then, we have

\begin{align*}
    \norm{r-\fl{r}}_1 \leq& 2^k \norm{(Dq)^2-(\fl{D} \fl{q})^2}_\infty \\
    \leq& 2^k \left(\norm{(Dq)^2-(\fl{D} \fl{q})(Dq)}_\infty + \norm{(\fl{D} \fl{q})(Dq)-(\fl{D} \fl{q})^2}_\infty \right)\\
    \leq& 2^k \left(\norm{Dq}_\infty\norm{Dq-\fl{D}\fl{q}}_\infty + \norm{\fl{D} \fl{q}}_\infty\norm{Dq-\fl{D}\fl{q}}_\infty \right) \\
    \leq& 2^k \left(\norm{q}_1+\norm{\fl{D}\fl{q}}_\infty \right)\norm{Dq-\fl{D} \fl{q}}_\infty \\
    \leq& 2^k \left(2\norm{q}_1+\norm{Dq-\fl{D}\fl{q}}_\infty \right)\norm{Dq-\fl{D} \fl{q}}_\infty\\
    \leq& 2^k \left(2\norm{q}_1+(c_\delta + 6k)\varepsilon \norm{q}_1 + c_1 \varepsilon^2 \right)\left((c_\delta + 6k)\varepsilon \norm{q}_1 + c_1 \varepsilon^2\right) \\
    \leq& 2^k \left(2(c_\delta + 6k)\varepsilon\norm{q}^2_1+(c_\delta + 6k)^2\varepsilon^2 \norm{q}_1^2 \right) + c_2\varepsilon^2 + c_3\varepsilon^3 + c_4\varepsilon^4 \\
    \leq& 2^k 2(c_\delta+6k)\varepsilon\norm{q}^2_1 + c_5\varepsilon^2
\end{align*}
where we have used the triangle inequality, the bound found above and absorbed the higher-order term in \(\varepsilon\) by an appropriate constant \(c_5\). We also have the following inequality
\[\norm{r}_1 = \sum_{i=0}^{2^k}[Dq(x_i)]^2 \leq \sum_{i=0}^{2^k}\norm{Dq}^2_\infty = 2^k\norm{Dq}^2_\infty \leq 2^k\norm{q}^2_1,\]
which we use in order to show that
\begin{align*}
    \norm{\fl{r}}_1 \leq& \sum_{i=0}^{2^k}\left(\abs{\fl{r(x_i)}-r(x_i)}+r(x_i)\right) \leq \norm{\fl{r} - r}_1 + \norm{r}_1 \\
    \leq& 2^k 2(c_\delta+6k)\varepsilon\norm{q}^2_1 + c_5\varepsilon^2 + 2^k\norm{q}_1^2 = 2^k[1+2(c_\delta + 6k)\varepsilon]\norm{q}_1^2 + c_5\varepsilon^2.
\end{align*}
Moving on, we have
\begin{align*}
    \norm{D^{-1}r-\fl{D^{-1}}\fl{r}}_\infty \leq& \norm{D^{-1}(r-\fl{r})}_\infty + \norm{(D^{-1}-\fl{D^{-1}})\fl{r}}_\infty \\
    \leq& \frac{\norm{r-\fl{r}}_1}{2^k} + \frac{6k\varepsilon \norm{\fl{r}}_1}{2^k}\\
    \leq& 2(c_\delta+6k)\varepsilon\norm{q}^2_1 + c_6\varepsilon^2 \\
    &\hspace{24.5mm} + 6k\varepsilon [1+2(c_\delta + 6k)\varepsilon]\norm{q}_1^2 + c_7\varepsilon^3 \\
    \leq& 2(c_\delta+9k)\varepsilon\norm{q}_1^2+c_8\varepsilon^2
\end{align*}
where again, in the first inequality, we have used the triangle inequality followed by the fact that \(\norm{D^{-1}x}_\infty \leq \frac{\norm{x}_1}{N},\, x \in \mathbb{R}^{2^k}\) and once again we use \cite[Lemma 4]{keich2005sfft} to proceed from the first to the second line. The transition to the last line is done by choosing an appropriate constant \(c_8\). Finally, we have that
\begin{align*}
    \norm{\fl{q^{\circledast 2}}-q^{\circledast 2}}_\infty =& \norm{\fl{h}\left(\fl{D^{-1}}\left[(\fl{D}\fl{q})^{2}\right]\right)-h\left(D^{-1}\left[(Dq)^{2}\right]\right)}_\infty \\ 
    =& \norm{\fl{h}\left(\fl{D^{-1}}\fl{r}\right)-h\left(D^{-1}r\right)}_\infty \\
    \leq& 2h(c_\delta+9k)\varepsilon\norm{q}_1^2+c_8h\varepsilon^2,
\end{align*}
for a suitable constant \(c_8\) that depends on \(c_\delta, k\) and \(\norm{q}_1^2\).
\end{proof}
The next Lemma shows that FFT-based convolution 
may be more sensitive to rounding errors, 
since we can only bound its absolute error. 
\begin{lemma}[Rounding error FFT-based convolution]\label{lem:roundingErrorFftBasedConv}
Let $\bar \alpha_N$ be computed by FFT-based convolution with \(n = 2^m, m\in \mathbb{N}\), let $\fl{\alpha_N}$ be computed with floating point arithmetic with machine epsilon $\varepsilon>0$. Assume further that \(\fl{f(x_j)} = f(x_j)\) for \(j \in \{0, 1, \hdots, N\}\) with \(N = 2^r, r \in \mathbb{N}\), and that \(2m \leq r\). Then, when disregarding higher-order epsilon terms, we have 
\[\norm{\fl{f^{\circledast 2^m}}-f^{\circledast 2^m}}_\infty \leq 18hc\log_2(nN)\log_2(n)\varepsilon  \norm{f}_1^n\]
where \(c = \max\{1, \gamma\}\).
\end{lemma}
\begin{proof}
Note that by applying recurrently the triangle inequality and by our assumption that \(\fl{f} = f\) we have
\[\norm{\fl{f^{\circledast 2^m}}-f^{\circledast 2^m}}_\infty \leq \sum_{i=0}^{m-1}\norm{\fl{\fl{f^{\circledast 2^{m-i-1}}}^{\circledast 2}}^{\circledast 2^i}-\left (\fl{f^{\circledast 2^{m-i-1}}}^{\circledast 2}\right)^{\circledast 2^i}}_\infty\]
For the sake of lighter notation we let \(g_i = \fl{f^{\circledast 2^i}}\), as we then can rewrite the right-hand side of the equation above as
\begin{equation}
\label{eq:bound_FFT_re_sum}
    \sum_{i=0}^{m-1}\norm{\fl{g_{m-i-1}^{\circledast 2}}^{\circledast 2^i}-\left (g_{m-i-1}^{\circledast 2}\right)^{\circledast 2^i}}_\infty.
\end{equation}
We now need to find an expression for each term in the sum above. First, for readability we introduce the shorthand notation \(\bar g \coloneqq g_{m-i-1}^{\circledast 2}\) for some arbitrary value of \(i \in \{1, \hdots, m-1\}\) and let \(\bar g\) be zero-padded such that \(\bar g \in \mathbb{R}^{2^{r+i}}\). Note then that
\begin{multline}
    \label{eq:D_fl_conv_est}
    \norm{D\fl{\bar g}}_\infty = \norm{D\fl{\bar g}-D\bar g+D\bar g}_\infty \leq \norm{D\fl{\bar g}-D\bar g}_\infty+\norm{D\bar g}_\infty \\
    \leq \norm{\fl{\bar g}-\bar g}_1+\norm{\bar g}_1. 
\end{multline}
Furthermore, we have that
\begin{align*}
\norm{(D\fl{\bar g})^{2^i}-(D\bar g)^{2^i}}_\infty \leq& \norm{(D\fl{\bar g})^{2^i}-(D\fl{\bar g})^{2^i-1}(D\bar g)}_\infty \\
&+ \norm{(D\fl{\bar g})^{2^{i-1}}(D\bar g)-(D\fl{\bar g})^{2^i-2}(D\bar g)^2}_\infty \\
&+ \cdots + \norm{(D\fl{\bar g})(D\bar g)^{2^{i-1}}-(D\bar g)^{2^i}}_\infty \\
\leq& \bigg(\norm{(D\fl{\bar g})^{2^i-1}}_\infty + \norm{(D\fl{\bar g})^{2^i-2}(D\bar g)}_\infty \\
&\hspace{1mm}+ \cdots + \norm{(D\fl{\bar g})(D\bar g)^{2^i-2}}_\infty \\
&\hspace{1mm}+ \norm{(D\bar g)^{2^i-1}}_\infty \bigg)\norm{D(\fl{\bar g}-\bar g)}_\infty \\
\leq& \bigg(\norm{D\fl{\bar g}}^{2^i-1}_\infty + \norm{D\fl{\bar g}}_\infty^{2^i-2}\norm{D\bar g}_\infty \\
&\hspace{1mm}+ \cdots + \norm{D\fl{\bar g}}_\infty\norm{D\bar g}^{2^i-2}_\infty + \norm{D\bar g}^{2^i-1}_\infty \bigg)\norm{\fl{\bar g}-\bar g}_1.
\end{align*}
Then from \eqref{eq:D_fl_conv_est} and Lemma \ref{lem:singelFFTConvBound} we have that 
\[\norm{D\fl{\bar g}}^s_\infty \leq (\norm{\fl{\bar g}-\bar g}_1 + \norm{\bar g}_1)^s \leq \left(18h(r+i+1)\varepsilon\norm{g_{m-i-1}}_1^2+ch\varepsilon^2 + \norm{\bar g}_1\right)^s.\]
Thus, we end up with \(\norm{D\fl{\bar g}}^s_\infty \leq \norm{\bar g}^s_1 + c\varepsilon\), where \(c\) captures all terms multiplied with \(\varepsilon\). Then we can write
\begin{align*}
\norm{(D\fl{\bar g})^{2^i}-(D\bar g)^{2^i}}_\infty \leq& \bigg(\norm{\bar g}_1^{2^i-1} + \norm{\bar g}^{2^i-1}_1 + \cdots \\
&+ \norm{\bar g}_1^{2^i-1} + \norm{\bar g}_1^{2^i-1} + c\varepsilon\bigg)2^i\norm{\fl{\bar g}-\bar g}_\infty \\
\leq& 2^i\norm{\bar g}_1^{2^i-1}2^i 18h(r+i+1)\varepsilon\norm{g_{m-i-1}}_1^{2} + c\varepsilon^2\\
\leq& 18 h 2^{2i} \norm{g_{m-i-1}}_1^{2^{i+1}-2} (r+i+1) \varepsilon\norm{g_{m-i-1}}_1^{2} + c\varepsilon^2 \\
\leq& 18\gamma (r+i+1)\varepsilon\norm{g_{m-i-1}}_1^{2^{i+1}} + c\varepsilon^2,
\end{align*}
where \(c\) is a constant capturing the higher-order terms in \(\varepsilon\) that we adjust appropriately from line to line and from our assumption that \(2m \leq r\) we have \(2^{2i}h \leq \gamma\) as well as the fact that \(\norm{g_{m-i-1}^{\circledast 2}}_1 \leq \norm{g_{m-i-1}}_1^2\). We can then show that 
\begin{align*}
    \norm{\fl{\bar g}^{\circledast 2^i}-\bar{g}^{\circledast 2^i}}_\infty =& \norm{h D^{-1}\left[(D\fl{\bar g})^{2^i}-(D\bar g)^{2^i}\right]}_\infty \\
    \leq& h\frac{\norm{(D\fl{\bar g})^{2^i}-(D\bar g)^{2^i}}_1}{2^{i+r}} \\
    \leq& h\norm{(D\fl{\bar g})^{2^i}-(D\bar g)^{2^i}}_\infty \\
    \leq& 18h\gamma(r+i+1)\varepsilon\norm{g_{m-i-1}}_1^{2^{i+1}} + c\varepsilon^2.
\end{align*}

We can now apply this inequality in order to get bounds on the terms in the sum in \eqref{eq:bound_FFT_re_sum}. Consider first the term in \eqref{eq:bound_FFT_re_sum} where \(i=0\), by applying Lemma \ref{lem:singelFFTConvBound} we have
\[\norm{\fl{g_{m-1}^{\circledast 2}}-\left (g_{m-1}^{\circledast 2}\right)}_\infty \leq 18h(r+1)\varepsilon\norm{g_{m-1}}_1^2.\]
Then, moving on to the case \(i \geq 1\) we have from the bound above that
\[
\norm{\fl{g_{m-i-1}^{\circledast 2}}^{\circledast 2^i}-\left (g_{m-i-1}^{\circledast 2}\right)^{\circledast 2^i}}_\infty \leq 18h\gamma (r+i+1)\varepsilon\norm{g_{m-i-1}}_1^{2^{i+1}} + c\varepsilon^2.
\]
By inserting the above estimates in \eqref{eq:bound_FFT_re_sum} we achieve the estimate
\begin{align*}
    \norm{\fl{f^{\circledast 2^m}}-f^{\circledast 2^m}}_\infty \leq& 18h(r+1)\varepsilon\norm{g_{m-1}}_1^2 \\
    &+ \sum_{i=1}^{m-1} 18h\gamma (r+i+1)\varepsilon\norm{g_{m-i-1}}_1^{2^{i+1}} + c\varepsilon^2 \\
    \leq& 18h(r+1)\varepsilon\norm{g_{m-1}}_1^2 \\
    &+ 18h\gamma (r+m)\varepsilon \sum_{i=1}^{m-1} \norm{g_{m-i-1}}_1^{2^{i+1}} + c\varepsilon^2.
\end{align*}
Thus, when only considering the leading term in \(\varepsilon\) we get
\begin{multline*}
    \norm{\fl{f^{\circledast 2^m}}-f^{\circledast 2^m}}_\infty \leq 18h\log_2(nN)\varepsilon \bigg(\norm{\fl{f^{\circledast 2^{m-1}}}}_1^2\\
    + \gamma\sum_{i=1}^{m-1} \norm{\fl{f^{\circledast 2^{m-i-1}}}}_1^{2^{i+1}}\bigg).
\end{multline*}

By now recursively applying this relation on the norms on the left hand side and ignoring higher-order terms in \(\varepsilon\) we are able to rewrite the equation above as
\begin{align*}
    \norm{\fl{f^{\circledast 2^m}}-f^{\circledast 2^m}}_\infty \leq& 18h\log_2(nN)\varepsilon \left(\norm{f}_1^{2^m} + \gamma\sum_{i=1}^{m-1} \norm{f}_1^{2^m}\right) \\
    \leq& 18hc\log_2(nN)\log_2(n)\varepsilon \norm{f}_1^n, \\
\end{align*}
where \(c = \max\{\gamma,1\}\).
\end{proof}

We are then ready to prove the following theorem, giving a bound on the error of performing convolution using FFT.
\begin{theorem}[Approximation error FFT-based convolution]
\label{thm:roundingErrorFFT}
Let the assumptions in Theorem~\ref{thm:errorEstimate} and Lemma~\ref{lem:roundingErrorFftBasedConv} hold. 
Then it holds that 
\begin{multline*}
|\fl{\bar \alpha_N} - \alpha| \le 
\gamma C_1 \overline{C}_2(n,\gamma) N^{-2p} + C_3 \max_{x \in [0,\gamma]} \left|\frac{d^r}{dx^r} f^{*n}(x)\right| N^{-r} \\
+ 18hc\log_2(nN)\log_2(n)\varepsilon \norm{f}_1^n).
\end{multline*}
\end{theorem}
\begin{proof}
    The result follows directly from Theorem~\ref{thm:errorEstimate} and Lemma~\ref{lem:roundingErrorFftBasedConv}.
\end{proof}

\subsection{Computational cost}

In this section we compare the computational cost and accuracy of 
direct-based convolution against FFT-based convolution 
as a function of the numerical resolution $N$ and the number of RVs $n$. We restrict ourselves to settings where the lemmas and theorems in Section~\ref{sec:theory} apply. 

\begin{theorem}\label{thm:costEstimate}
  The computational cost of computing $\bar \alpha_N$, counted in the number of floating point operations,
  is
  \[
    \text{COST}(\bar \alpha_N) =
    \begin{cases}
      \mathcal{O}(\log_2(n) N^2) &      \text{when using direct convolution}\\
      \mathcal{O}(\log_2(n) N \log_2(N)) & \text{when using FFT-based convolution.} 
      \end{cases}
  \]
 \end{theorem}

 \begin{proof}
   Recall that $m$ denotes the largest integer such that
  $m \le \log_2(n)$.  For each $\ell =1,\ldots,m$ and
  $k =0,1,\ldots, N$, the computation
  $\bar f^{\circledast 2^{\ell}}(x_k) = \bar f^{\circledast 2^{\ell-1}} \circledast \bar
  f^{\circledast 2^{\ell-1}}(x_k)$ costs $\mathcal{O}(N)$. The cost of computing
  $\bar f^{\circledast n}$ thus becomes $\mathcal{O}( m N^2)$ and computing
  $\bar \alpha$ adds an additional (relatively speaking, negligible)
  cost of $\mathcal{O}(N)$. The upper bound in cost for FFT-based convolution follows by a similar argument. 
  \end{proof}

When disregarding factors of $\log_2(N)$ in the cost estimate
and also disregarding rounding errors, we obtain the following relation between cost and absolute approximation error:
\begin{equation}\label{eq:costVsErrorNonRelative}
|\bar \alpha_N - \alpha | \le 
   \begin{cases} 
    \frac{\widehat C}{(\text{COST})^{r/2}} & \text{when using direct convolution}\\
    \frac{\widehat C}{(\text{COST})^{r}} & \text{when using FFT-based convolution.} 
   \end{cases}
\end{equation}
Supposing further that there exists a constant $C>0$ such that 
\[
\frac{ (1 + 4 n N \epsilon) \left(\gamma C_1 \overline{C}_2(n,\gamma) N^{r-2p} + C_3 \max_{x \in [0,\gamma]} \left|\frac{d^r}{dx^r} f^{*n}(x)\right|\right) }{\alpha}  \le C 
\]
holds for all relevant $\gamma, n, N$ and $\varepsilon$, we obtain  
the following error estimate for the relative error of approximating $\alpha$: 
\begin{equation}\label{eq:costVsError}
\frac{|\bar \alpha_N - \alpha |}{\alpha} \le 
   \begin{cases} 
    C N^{-r} + nN \varepsilon & \text{for direct convolution}\\
    C N^{-r} + C_6\frac{hc \log_2(nN)\log_2(n) \norm{f}_1^{n}\varepsilon}{\alpha} & \text{for FFT-based convolution}. 
   \end{cases}
\end{equation}
We note that when $\alpha \ll 1$, 
the result indicates that for a given resolution $N$, the relative error may be 
substantially smaller for direct-based convolution than for FFT-based convolution, precisely as we observe in the numerical examples in Section~\ref{sec:numerics}.

\section{Numerical experiments}\label{sec:numerics}

To verify numerically that the proposed method produces satisfactory results and to confirm that the theoretical error rate identified in the previous section holds in practice, we conducted a series of experiments. First, in Section \ref{sec:convVsFFT} we compare the FFT implementation of the convolution method with the direct method in terms of how well they are able to approximate the rare-event probability of a sum of RVs. As the results from the first experiment shows that the direct method gives low rounding errors, we run the rest of the experiments using the direct method only. In Section \ref{sec:estKnownDistr}, we look at how well the convolution method estimate the CDF for the sum of RVs for which the distribution of the sum is indeed known. Then, in Section \ref{sec:Lognormal}, we consider the Log-Normal distribution with two goals in mind: 1) We want to explore the convergence properties of the convolution method and check if we empirically are able to observe the theoretical convergence rate as given by Theorem \ref{thm:errorEstimate}, 2) We compare the calculated estimates of the CDF with approximations calculated using an alternative method, in this instance a saddlepoint method presented in \cite{Asmussen16}. Then, in the last Section \ref{sec:estUnknownDistr} we look at how the convolution method performs when approximating the CDF for the sum of RVs for other distributions where the distribution of the sum is not known.

\subsection{Comparison of direct- and FFT-based convolution}\label{sec:convVsFFT}
In this section, we compare the performance of direct convolution and FFT-based 
convolution for left-tail rare-event estimation. In agreement with the theoretical results in Section~\ref{sec:roundingErrors}, we show that FFT-based convolution is more
sensitive to rounding errors than direct convolution in two problem settings where $\alpha \ll1$. 

\subsubsection{Log-Normal distribution}
We estimate the probability of $Y = \sum_{i=1}^{16} X_i \le \gamma$, where $X_i$ for $i=1,2,\cdots,16$ are i.i.d $\text{Log-Normal}(0,1/64)$ with density denoted by $f$. The large variation in magnitude for the density of $Y$ is illustrated in the left plot of 
Figure~\ref{fig:l_pdfFFT_vs_conv_r_pdfFFT_vs_conv_runtime}, where we numerically have computed $p(y):= \bar{f}^{\circledast 16}(y)$ over the interval $y \in [8,16]$ using $N=10^6$ quadrature points. The 
density is computed by direct convolution with Matlab's \textbf{conv()} function and 64-bit floating point precision, and by the FFT-based method for a range of different floating point precisions, using the multiple precision toolbox~\cite{advanpix2006multiprecision}. We observe that the higher the precision, the better the FFT-based convolution approximates the direct convolution's density, and that FFT introduces an approximation error 
that is proportional to the machine epsilon.
This is consistent with the observations in~\cite{wilson2017accurate}.
For reference, we note that the machine epsilon is approximately $1.19 \times 10^{-7}$ for 32-bit floats, $2.22 \times 10^{-16}$ for 64-bit floats, $1.93 \times 10^{-34}$ for 128-bit floats, and $1.81\times 10^{-71}$ for 256-bit floats. 
\begin{figure}
\center
  \includegraphics[width=0.49\linewidth]{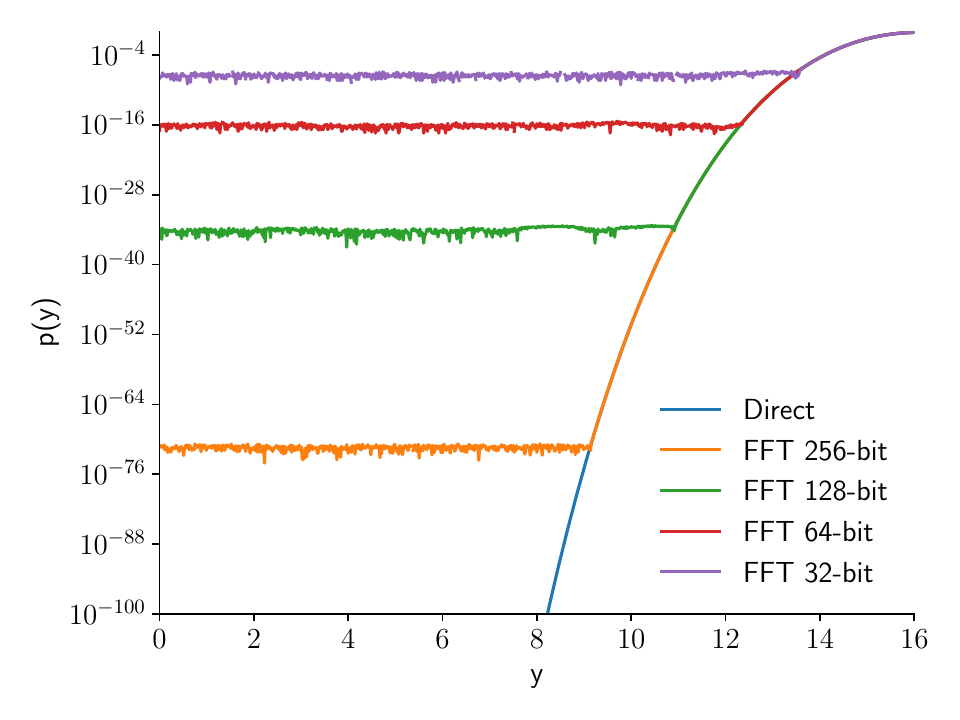}
  \includegraphics[width=0.49\linewidth]{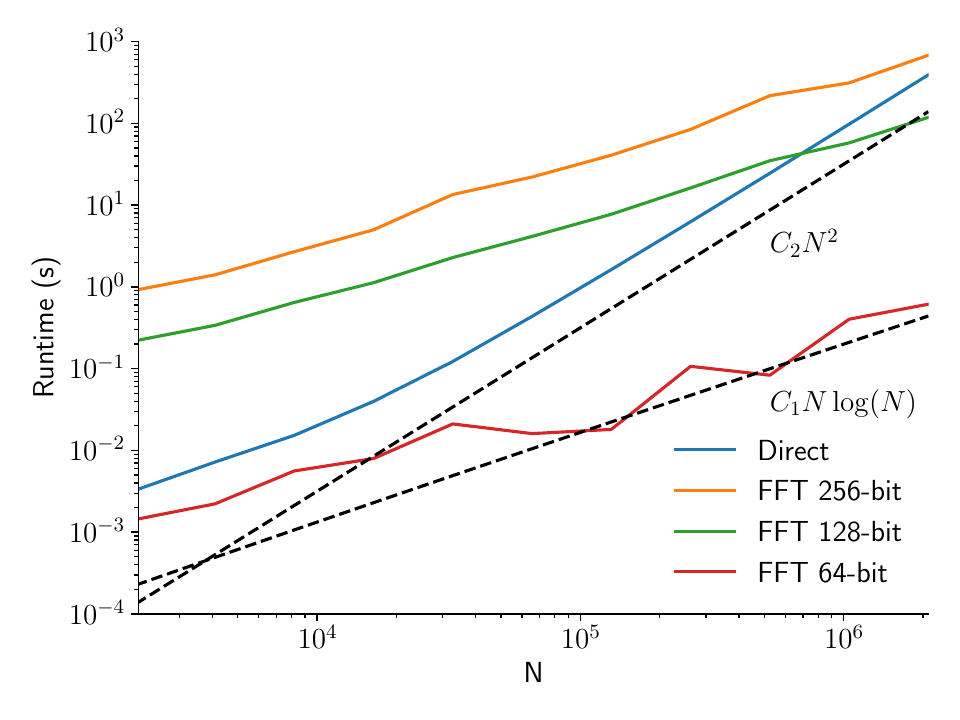}
  \caption{Left: The probability density function $p(y) = \bar f^{\circledast 16}(y)$ for direct convolution and FFT-based convolution for the rare-event problem studied 
  in Section~\ref{sec:convVsFFT}. Right: The runtime for direct convolution and FFT-based convolution for the rare-event problem studied 
  in Section~\ref{sec:convVsFFT}.}
  \label{fig:l_pdfFFT_vs_conv_r_pdfFFT_vs_conv_runtime}
\end{figure}
 Table~\ref{tab:convVsFFT_lognormal} presents the relative error $|\alpha - \fl{\bar \alpha_N}|/\alpha$ for different values of $\gamma$ using 64-bit precision direct convolution and FFT-based convolution for a range of precisions. All methods use $N=10^6$ quadrature points. The pseudo-reference solution is computed using $N=2^{21}$ quadrature points with 512-bit precision FFT-based convolution. The FFT-based convolution only approximates the rare event well when $\alpha$ is much larger than the machine epsilon for the floating point precision. 

\begin{table}
\scriptsize
\centering
  \begin{tabular}{c c c c c c c c}
    \hline
    $\gamma$    & Ref.~sol.~CDF           & Dir.~conv. & Saddlp. & FFT 32-bit & FFT 64-bit & FFT 128-bit & FFT 256-bit         \\
    \hline
8.8 & 2.05$\times10^{-83}$ & 4.99$\times10^{-13}$ & 4.90$\times10^{-06}$ & 4.92$\times10^{+75}$ & 1.67$\times10^{+67}$ & 3.58$\times10^{+49}$ & 1.48$\times10^{+11}$\\
9.6 & 1.02$\times10^{-61}$ & 5.77$\times10^{-13}$ & 5.46$\times10^{-06}$ & 5.34$\times10^{+53}$ & 1.27$\times10^{+45}$ & 1.25$\times10^{+28}$ & 1.01$\times10^{-10}$\\
10.4& 1.04$\times10^{-44}$ & 5.47$\times10^{-13}$ & 5.61$\times10^{-06}$ & 1.66$\times10^{+37}$ & 2.23$\times10^{+28}$ & 1.02$\times10^{+11}$ & 2.24$\times10^{-28}$\\
11.2& 1.76$\times10^{-31}$ & 6.00$\times10^{-13}$ & 5.24$\times10^{-06}$ & 1.74$\times10^{+24}$ & 2.94$\times10^{+15}$ & 1.56$\times10^{-03}$ & 1.46$\times10^{-40}$\\
12  & 2.45$\times10^{-21}$ & 5.81$\times10^{-13}$ & 4.27$\times10^{-06}$ & 4.83$\times10^{+13}$ & 3.20$\times10^{+04}$ & 1.91$\times10^{-13}$ & 1.16$\times10^{-50}$\\
12.8& 9.81$\times10^{-14}$ & 5.74$\times10^{-13}$ & 2.72$\times10^{-06}$ & 1.96$\times10^{+05}$ & 2.37$\times10^{-03}$ & 1.87$\times10^{-20}$ & 3.80$\times10^{-59}$\\
13.6& 3.03$\times10^{-08}$ & 5.97$\times10^{-13}$ & 6.88$\times10^{-07}$ & 5.38$\times10^{+00}$ & 1.23$\times10^{-08}$ & 3.93$\times10^{-26}$ & 1.79$\times10^{-63}$\\
14.4& 1.63$\times10^{-04}$ & 6.02$\times10^{-13}$ & 1.66$\times10^{-06}$ & 3.01$\times10^{-03}$ & 4.37$\times10^{-13}$ & 1.09$\times10^{-29}$ & 5.45$\times10^{-68}$\\
\hline   
  \end{tabular}
  \caption{Comparison of the relative error $|\fl{\bar \alpha_N} - \alpha|/\alpha$ for sums of Log-Normal-distributed RVs using 64-bit direct convolution, the saddlepoint method (computed with 512-bit floating bit precision) and FFT-based convolution computed with four different floating point precisions.}
  \label{tab:convVsFFT_lognormal}
\end{table}

 Our numerical studies indicate that FFT-based convolution is much more sensitve to rounding errors than direct convolution, but it may still be useful in computations when the number of quadrature points $N$ is large since the method has a lower asymptotic computational cost than direct convolution, cf.~Theorem~\ref{thm:costEstimate}. The right plot in Figure~\ref{fig:l_pdfFFT_vs_conv_r_pdfFFT_vs_conv_runtime} measures the 
 computational cost of the two convolution methods in runtime, displays cost 
 rates that are consistent with Theorem~\ref{thm:costEstimate}, and shows that for all considered floating point precisions, FFT-based convolution will eventually, for sufficiently large $N$, outperform direct convolution in terms of runtime. 

\subsubsection{L{\'e}vy distribution}

We next estimate the probability of $Y = \sum_{i=1}^{16} X_i \le \gamma$, with $X_i$, $i=1,2,\cdots,16$ are i.i.d. $\text{L{\'e}vy}(0, 0.1)$ whose density is denoted by $f$ (given in Table~\ref{tab:density_known}). This is a stable distribution, which is very suitable for validation of the numerical methods since its
density and CDF are known: $Y \sim \text{L{\'e}vy}(0, 25.6)$ 
and $\alpha = \mathbb{P}(Y \le \gamma) = \text{erfc}( \sqrt{12.8/\gamma})$.
The large variation in magnitude for the density of $Y$ is illustrated in Figure~\ref{fig:pdfDirConvVsFFT_levy}, where we compare the exact density of $Y$ to numerical approximations using $N=10^6$ quadrature points.
\begin{figure}[h!]
\centering
  \includegraphics[width=0.58\linewidth]{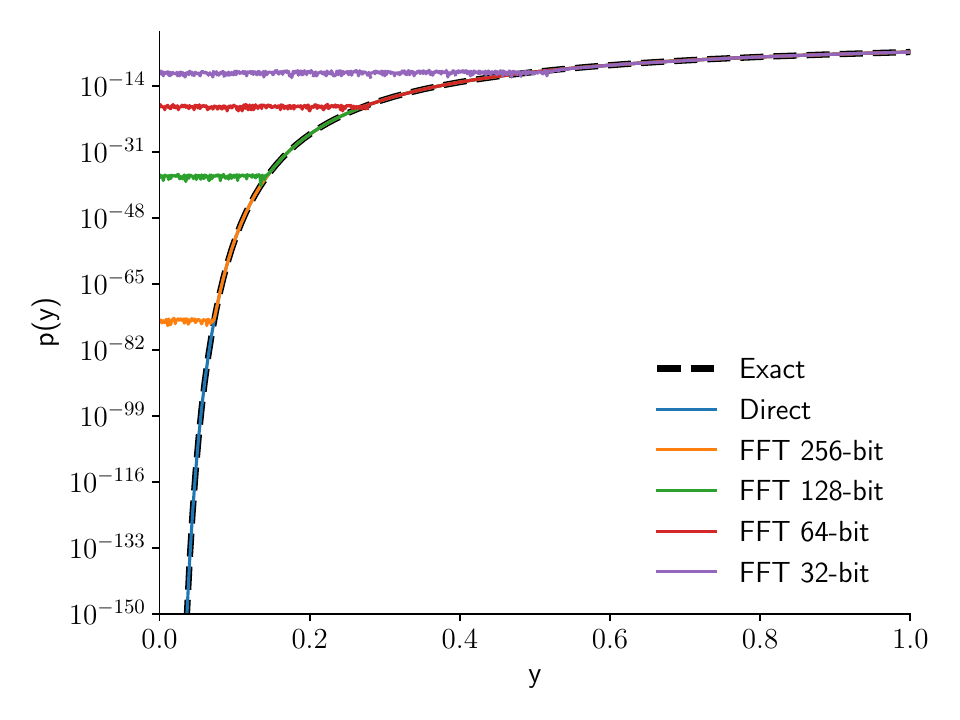}
  \caption{The probability density function 
  to $Y = \sum_{i=1}^{16}X_i$ where $X_i$
  are i.i.d.~L\'evy distributed RVs. The density is computed by the exact formula, by direct convolution and FFT-based convolution.}
  \label{fig:pdfDirConvVsFFT_levy}
\end{figure}
The density is computed numerically by direct convolution using 64-bit precision and FFT-based convolution for a range of floating point precisions. Remarkably, direct convolution agrees well with the exact density over the full range of values displayed, while FFT-based convolution agrees well only when the precision is sufficiently high. 

Table~\ref{tab:convVsFFT_levy} presents the relative error $|\alpha - \fl{\bar \alpha_N}|/\alpha$ for different values of $\alpha$ using 64-bit precision direct convolution and FFT-based convolution for a range of precisions. All numerical methods use $N=10^6$ quadrature points. The reference solution is computed by the exact CDF, $\alpha = \text{erfc}( \sqrt{12.8/\gamma})$ with 512-bit floating point precision. We again observe that FFT-based convolution only approximates the rare event well when $\alpha$ is much larger than the floating point precision machine epsilon.
\begin{table}
\small
\centering
  \begin{tabular}{c c c c c c c }
    \hline
    $\gamma$    & $\alpha=\mathbb{P}(Y \le \gamma)$   & Dir.~conv.  & FFT 32-bit & FFT 64-bit & FFT 128-bit & FFT 256-bit         \\
    \hline
0.05 & 2.33$\times10^{-113}$ & 6.74$\times10^{-13}$ & 7.82$\times10^{+68}$ & 6.12$\times10^{+50}$ & 2.16$\times10^{+33}$ & 1.58$\times10^{-04}$\\
0.1 & 1.28$\times10^{-57}$ & 6.78$\times10^{-13}$ & 7.45$\times10^{+25}$ & 6.05$\times10^{+15}$ & 3.58$\times10^{-02}$ & 7.47$\times10^{-27}$\\
0.2 & 1.12$\times10^{-29}$ & 6.05$\times10^{-13}$ & 1.15$\times10^{+09}$ & 2.61$\times10^{-02}$ & 3.45$\times10^{-18}$ & 9.20$\times10^{-29}$\\
0.5 & 8.34$\times10^{-13}$ & 4.24$\times10^{-13}$ & 2.68$\times10^{-02}$ & 1.51$\times10^{-10}$ & 6.18$\times10^{-28}$ & 1.54$\times10^{-31}$\\
1 & 4.20$\times10^{-07}$ & 2.80$\times10^{-13}$ & 4.95$\times10^{-05}$ & 1.24$\times10^{-13}$ & 2.22$\times10^{-32}$ & 1.11$\times10^{-34}$\\
\hline   
  \end{tabular}
  \caption{Comparison of the relative error $|\fl{\bar \alpha_N} - \alpha|/\alpha$ for sums of L\'evy-distributed 
  RVs using 64-bit direct convolution and FFT-based convolution computed with four different floating point precisions.}
  \label{tab:convVsFFT_levy}
\end{table}

\subsection{Estimating known distributions}
\label{sec:estKnownDistr}
There exist several probability distributions for which the distribution of the sum $Y = \sum_{i=1}^nX_i$ is known, given that the RVs $X_i$ are all independent. If for example $X_i$ is Chi-squared distributed with $r_i$ degrees of freedom ($X_i \sim \chi^2(r_i)$) for $i \in \{1, 2, \hdots, n\}$ we have that
\[Y \sim \chi^2 \left (\sum_{i=1}^n r_i\right ).\]
Thus, to check empirically that the algorithm presented in this paper indeed is able to accurately calculate the PDF and the CDF of the sum of RVs we check against distributions where the resulting distribution is known. We chose to do numerical experiments with the Chi-squared ($\chi^2$) and Lévy distributions. The PDFs are given in Table \ref{tab:density_known}. In the case of the Lévy distribution if we let $X_i \sim \text{Lévy}(\mu_i, c_i)$ with $\mu_i \in (-\infty, \infty), c_i > 0$ for $i \in \{1, 2, \hdots, n\}$ we have that
\[Y \sim \text{Lévy} \left (\sum_{i=1}^n\mu_i, \left(\sum_{i=1}^n \sqrt{c_i} \right)^2 \right ).\]

\begin{table}
    \centering
    \begin{tabular}{ccc}
    \hline
    Distribution            & Parameters         & PDF \\
    \hline
    Chi-squared $(\chi^2)$  & $df \in\mathbb{N}$  & $\frac{1}{2^{df/2} \Gamma(df/2)}x^{df/2-1}e^{-x/2}$\\
    Lévy                    & $c > 0$            & $\sqrt{\frac{c}{2\pi}}\frac{e^{-\frac{c}{2x}}}{x^{3/2}}$\\
    \hline
    \end{tabular}
    \caption{Probability density function for the Chi-squared and Lévy distributions}
    \label{tab:density_known}
\end{table}

For these experiments we are estimating the value $\alpha = F_Y(\gamma), \gamma = xn$ with $x = 0.05$ and $n = 16$, i.e.
\[\alpha = F_Y(0.8) = \mathbb{P}(Y \leq 0.8) = P \left (\sum_{i=1}^{16} X_i \leq 0.8 \right )\]
where $F_Y$ is the CDF of $Y = \sum_{i=1}^{16} X_i$ with the RVs $X_i, i \in \{1, 2, \hdots, 16\}$ all independent. The estimates $\bar \alpha$ are calculated using equation \eqref{eq:newtonCotesAlpha} with Boole's rule as the closed Newton-Cotes formula in the last step (see Table \ref{tab:newton_cotes_formulas} for an overview of the weights).
\begin{figure}
\center
  \includegraphics[width=0.49\linewidth]{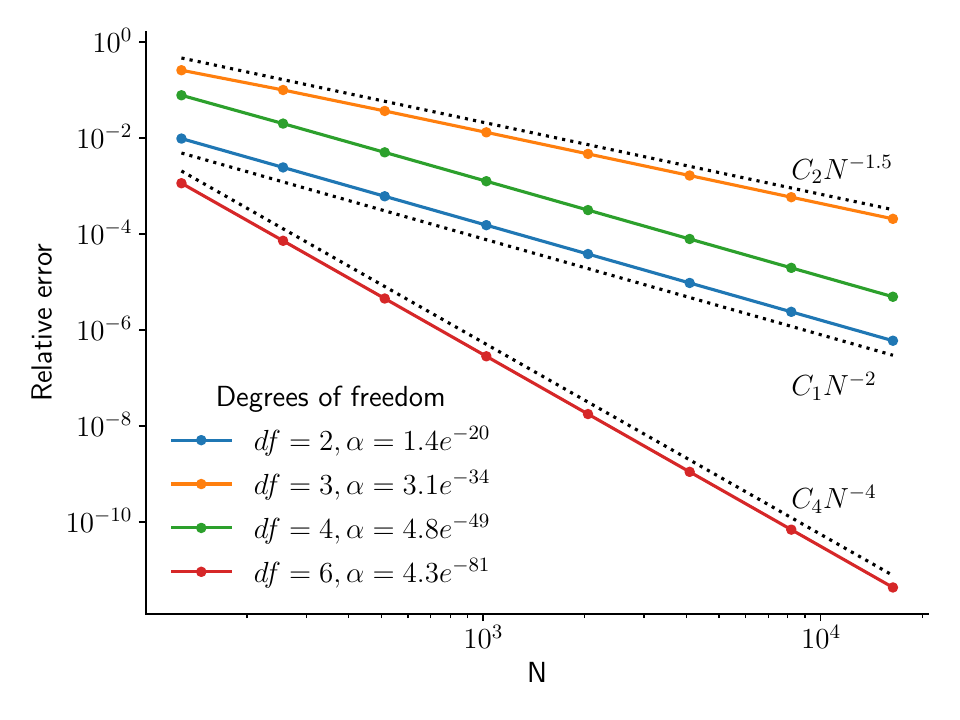}
  \includegraphics[width=0.49\linewidth]{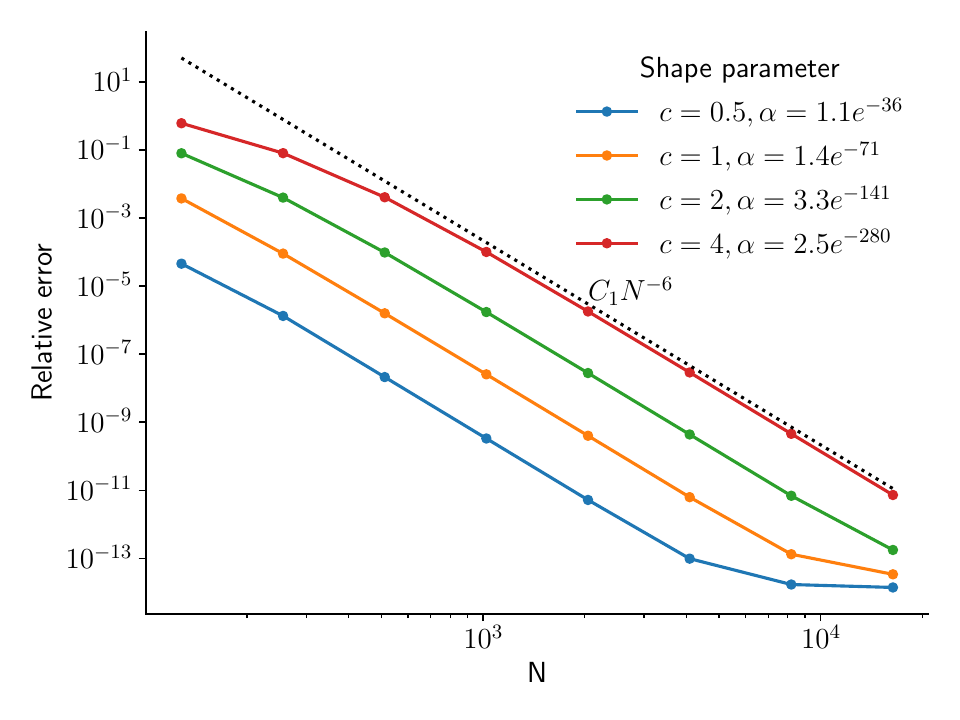}
  \caption{Left: Relative error as a function of the mesh-size when estimating $F_Y(0.8), Y \sim \chi^2(16df)$. Right: Relative error as a function of the mesh-size when we are estimating $F_Y(0.8), Y \sim \text{Lévy} \left (0, \left(\sum_{i=1}^n \sqrt{c_i} \right)^2 \right )$}
  \label{fig:l_chi_squared_r_levy}
\end{figure}
In the left plot of Figure \ref{fig:l_chi_squared_r_levy} we display the relative error $\delta = \frac{\abs{\alpha - \bar \alpha}}{\alpha}$ as a function of the mesh size $N$ when estimating $\alpha$ with $X_i \sim \chi^2(df)$ for a number of different parameter values $df$. Note that the legend also display the value of \(\alpha\), showing that we indeed are estimating rare events. It is apparent from the figure that there is a large difference in the convergence rate depending on the value of $df$, with the convolution method converging faster to the correct value of $\alpha$ when the value of $df$ increase. It is straight forward to check that \(f_3'(x) \xrightarrow{x \to 0}\infty\), \(f_2'(0) = c_1\), \(f_4'(0) = c_2\) and \(f_6'(0) = 0\) for some constants \(c_1, c_2 \in \mathbb{R}\) and where \(f_{df}\) is the PDF for the \(\chi^2\)-distribution with parameter value \(df\). Note also that \(f_6''(0) = c_3\) for some \(c_3 \in \mathbb{R}\). The observed convergence rates are therefore coherent with the theory, with the notable exception of the case \(df = 3\). We suspect that the bad convergence rate is a result of \(f_3'(x) \xrightarrow{x \to 0}\infty\) . Note also that \(f_2(0) \neq 0\) thus our implementation of the convolution method used for this experiment utilize the formula given in \ref{eq:discLinConv2}.

The results of a similar experiment where $X_i \sim \text{Lévy}(0, c)$ for a number of different values for the shape parameter $c$ is shown in the right plot in Figure \ref{fig:l_chi_squared_r_levy}. From the figure we see that the convolution method performs really well when estimating the Lévy distribution, with the relative error $\delta$ being less than $10^{-9}$ for all values of $c$ that we tested when the mesh-size \(N > 10^4\). Furthermore, we observe from the figure that a convergence rate of about \(N^{-6}\) were attained for all choices of \(c\). In contrast with the \(\chi^2\)-distribution we have \(f^{(k)}(0) = 0\) for all \(k \in \mathbb{N}\) for our choices of \(c\). The result in Theorem \ref{thm:errorEstimate} therefore implies that we should observe a convergence rate of at least \(N^{-6}\) as we use Boole's rule for the integration in the last step. This fits well with the observed convergence rate for all choices of \(c\). Furthermore, we observe that the relative error flattens out when we reach errors bellow \(10^{-13}\), which is probably due to round-off errors and in agreement with with the errors observed in Section \ref{sec:convVsFFT}.

\subsection{Convergence properties for the convolution method}
\label{sec:Lognormal}
In this subsection we again consider the problem of estimating the value
\[\alpha = F_Y(\gamma)\]
where $\gamma = xn$ with $x = 0.7$ and $n = 16$, i.e.
\[\alpha = F_Y(11.2) = \mathbb{P}(Y \leq 11.2) = \mathbb{P} \left (\sum_{i=1}^{16} X_i \leq 11.2 \right ),\]
where $F_Y$ is the CDF of $Y = \sum_{i=1}^{16} X_i$ and $X_i \sim \text{Log-Normal}(0, \sigma^2), i \in \{1, \hdots, 16\}$ are i.i.d. Here, we let $\sigma = 0.125$. The PDF is given in Table \ref{table1}. The approximations $\bar\alpha$ are calculated using equation \eqref{eq:newtonCotesAlpha}. In this experiment we aim to observe how the convergence rate vary when utilize three different closed Newton-Cotes formulas in the last step, Trapezoid, Simpson and Boole. The weights \(w_j\) in equation \eqref{eq:newtonCotesAlpha} depend on the used Newton-Cotes formula. The resulting formula for each of the Newton-Cotes formulas utilized in this experiment is given in Table \ref{tab:newton_cotes_formulas}. We also compare our estimates of the CDF value \(\alpha\) and the PDF value \(f(11.2)\) with approximations generated by a saddlepoint method presented in \cite{Asmussen16}. For details on the saddlepoint method we refer the reader to the cited paper. Furthermore, the estimate of \(f(11.2)\) is numerically found by the value \(f^{\circledast n}(x_N)\). 

\begin{table}
    \centering
    \begin{tabular}{cp{8cm}}
    \hline
    Newton-Cotes formula & Formula \\
    \hline
    Trapezoid & $h \left (\frac{1}{2} (\bar f^{\circledast n}(x_0) + \bar f^{\circledast n}(x_N)) + \sum_{j \in \{1, 2, \hdots, N-1\}}\bar f^{\circledast n}(x_j) \right )$\\
    Simpson   & $\frac{h}{3} \Big (\bar f^{\circledast n}(x_0) + \bar f^{\circledast n}(x_N) + 4\sum_{j \in \{1, 3, \hdots, N-1\}} \bar f^{\circledast n}(x_j) \newline
    \hbox{\hspace{4mm}}+ 2\sum_{j \in \{2, 4, \hdots, N-2\}}\bar f^{\circledast n}(x_j) \Big )$\\
    Boole's   & $\frac{2h}{45} \Big (7\big(\bar f^{\circledast n}(x_0) + \bar f^{\circledast n}(x_N)\big) \newline
    \hbox{\hspace{5mm}}+ 32 \sum_{j \in \{1, 3, \hdots, N-1\}} \bar f^{\circledast n}(x_j) \newline
    \hbox{\hspace{5mm}}+ 12 \sum_{j \in \{2, 6, \hdots, N-2\}} \bar f^{\circledast n}(x_j) \newline
    \hbox{\hspace{5mm}}+ 14 \sum_{j \in \{4, 8, \hdots, N-4\}} \bar f^{\circledast n}(x_j) \Big )$\\
    \hline
    \end{tabular}
    \caption{Explicit Newton-Cotes formulas}
    \label{tab:newton_cotes_formulas}
\end{table}

In order to test the convergence rate using the different rules in the last step we apply an iterative scheme where we first calculate a pseudo-reference solutions $\bar \alpha_{N_M}$ using Boole's rule on a mesh of size $N_M$ for some large $N_M$. We then calculate estimates $\bar \alpha_{N_m}$ using the three different rules on a mesh of size $N_m$, where $N_m \ll N_M$, and calculate the relative error
\[\bar \delta = \frac{\abs{\bar\alpha_{N_M} - \bar \alpha_{N_m}}}{\alpha_{N_M}}\]
for each of the estimates. We then check if all of the calculated relative errors are smaller than some threshold $\epsilon > 0$. If this is not the case we repeat the calculation on a mesh of size $N_m' = 2N_m$. This is repeated until all calculated errors are bellow the given threshold. Here we set $N_M = 2^{17}, N_m = 2^{10}$ and $\epsilon = 10^{-8}$. 

The results are shown in the left plot of Figure \ref{fig:convLognorm} together with reference slopes showing the theoretical convergence rates and the relative error of the above mentioned saddlepoint method. Given the fact that for the PDF \(f\) of the Log-Normal distribution we have \(f^{(k)}(0) = 0 \) for all \(k \in \mathbb{N}\), we would from Theorem \ref{thm:errorEstimate} expect convergence rates similar to the convergence rate of the chosen Newton-Cotes formula. Our results are in accordance with this expectation as the convergence rates are \(2, 4\) and \(6\) for the Trapezoid, Simpson and Boole's respectively. We also see that we quickly achieve relative errors smaller than the saddlepoint method with all three rules.

\begin{figure}
\center
  \includegraphics[width=0.49\linewidth]{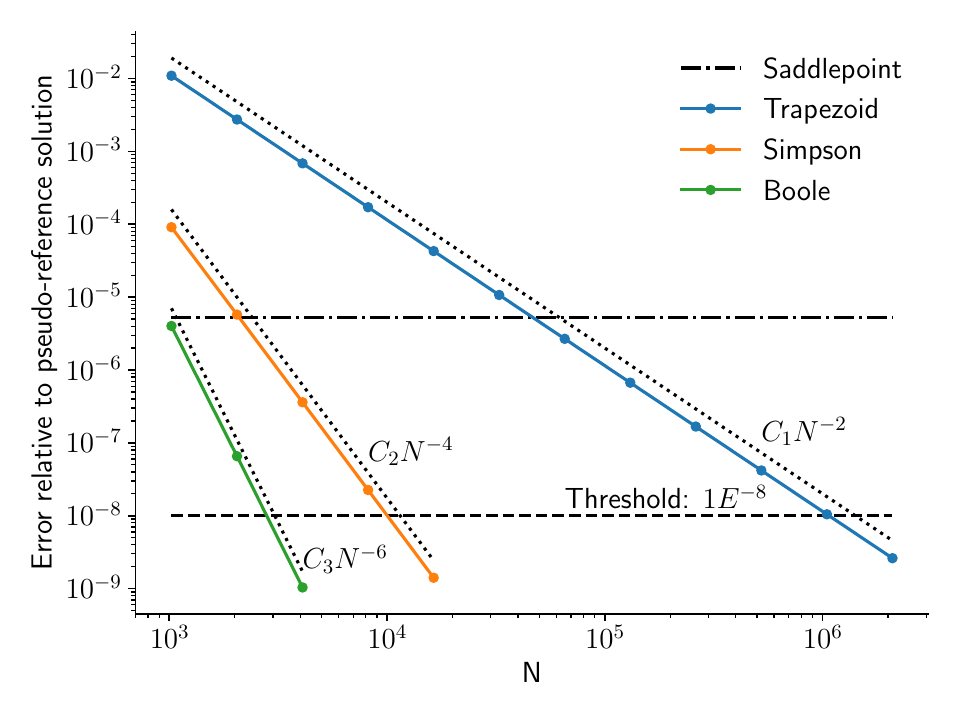}
  \includegraphics[width=0.49\linewidth]{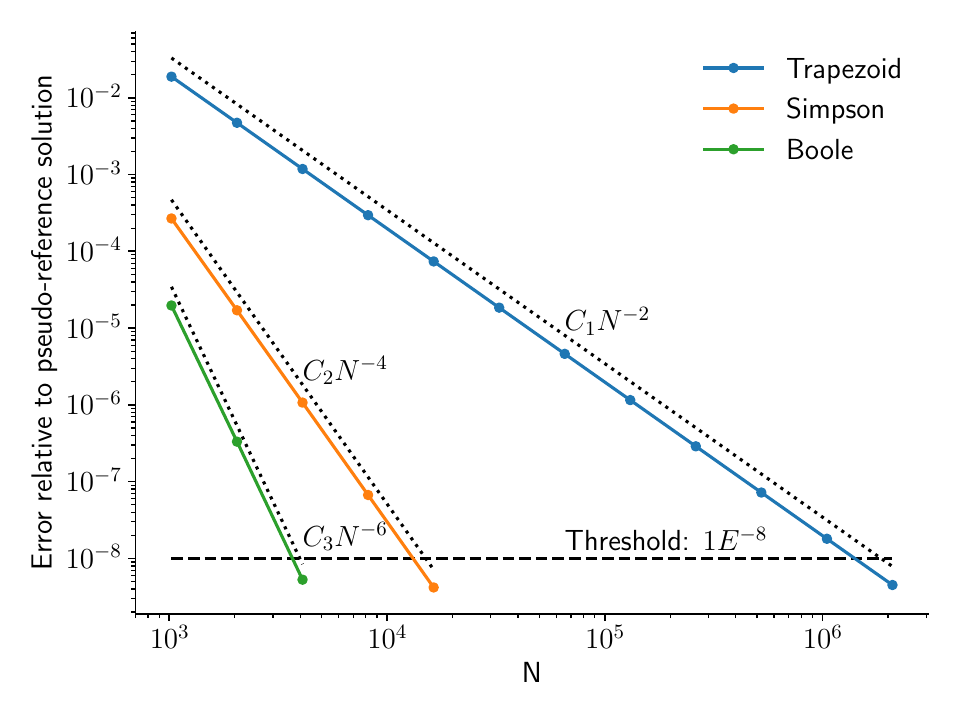}
  \caption{Left: Relative error as a function of the mesh-size when approximating $\alpha_{N_M}$ in the case when \(X_i \sim \text{Log-Normal}(0, 0.125)\) where (\(\alpha_{N_M}\) is a pseudo-reference solution calculated using the convolution method with $N_M = 1e6$ using Boole's rule in the last step. Right: Relative error as a function of the mesh-size when estimating pseudo-reference solution $\alpha_{N_M}$ of the CDF of a sum of Log-Normals with varying $\sigma$ calculated using the convolution method with $N_M = 1e6$ using Boole's rule in the last step}
  \label{fig:convLognorm}
\end{figure}

In order to further validate the correctness of the convolution method when applied to the Log-Normal distribution, we compared estimates of the CDF and PDF for different values of \(x\) generated with our method with approximations calculated by the saddlepoint method from \cite{Asmussen16}. For the convolution method we utilized a mesh-size of $10^4$ and Boole's rule in the last step, while the saddelpoint approximations were calculated using our implementation of the saddlepoint method presented in \cite{Asmussen16}. The results are given in \ref{tab:resultsLognormalSum}. Note that the values calculated with our implementation of the saddlepoint method gives slightly different values compared with the values listed in the paper \cite{Asmussen16} (\(\pm 1\) in the third digit of the significand).

The results show that the two methods gives similar approximations of $\alpha$ for all chosen values of $x$, with only the fourth non-zero decimal being different between the two estimates in some cases. This indicates that the convolution method is indeed able to accurately estimate the desired probabilities. Both our implementation of the convolution method and the saddlepoint method has for this experiment a negligible CPU-time. The advantage with the convolution method as opposed to the more intricate saddlepoint method is that the convolution method is generic in the way that it can handle multiple distributions and that the RVs does not need to be identically distributed. However, for the convolution method when a larger mesh size $N$ is needed to get accurate estimates the cost increases by $\mathcal{O}(N^2)$ and when the number $n$ of RVs in the sum increases the complexity increase by $\mathcal{O}(\log_2(n))$. On the other hand, the saddlepoint method has a negligible computational cost.

\begin{table}
\center
  \begin{tabular}{c c c c c}
    \hline
    x    & Convolution CDF           & Saddle CDF                & Convolution PDF           & Saddle PDF             \\
    \hline
    0.70 & 1.761 $\times 10^{-31}$ & 1.761 $\times 10^{-31}$ & 5.873 $\times 10^{-30}$ & 5.873 $\times 10^{-30}$ \\
    0.80 & 9.806 $\times 10^{-14}$ & 9.806 $\times 10^{-14}$ & 1.829 $\times 10^{-12}$ & 1.829 $\times 10^{-12}$ \\
    0.85 & 3.031 $\times 10^{-8}$  & 3.031 $\times 10^{-8}$  & 3.975 $\times 10^{-7}$  & 3.975 $\times 10^{-7}$  \\
    0.90 & 1.631 $\times 10^{-4}$  & 1.631 $\times 10^{-4}$  & 1.388 $\times 10^{-3}$  & 1.388 $\times 10^{-3}$  \\
    0.91 & 5.955 $\times 10^{-4}$  & 5.955 $\times 10^{-4}$  & 4.577 $\times 10^{-3}$  & 4.577 $\times 10^{-3}$  \\
    0.92 & 1.911 $\times 10^{-3}$  & 1.911 $\times 10^{-3}$  & 1.318 $\times 10^{-2}$  & 1.318 $\times 10^{-2}$  \\
    0.93 & 5.423 $\times 10^{-3}$  & 5.423 $\times 10^{-3}$  & 3.332 $\times 10^{-2}$  & 3.332 $\times 10^{-2}$  \\
    0.94 & 1.368 $\times 10^{-2}$  & 1.368 $\times 10^{-2}$  & 7.416 $\times 10^{-2}$  & 7.416 $\times 10^{-2}$  \\
    0.95 & 3.081 $\times 10^{-2}$  & 3.081 $\times 10^{-2}$  & 1.460 $\times 10^{-1}$  & 1.460 $\times 10^{-1}$  \\
    0.98 & 1.901 $\times 10^{-1}$  & 1.901 $\times 10^{-1}$  & 5.520 $\times 10^{-1}$  & 5.520 $\times 10^{-1}$  \\
    \hline    
  \end{tabular}
  \caption{Approximations of the CDF and PDF of Y}
  \label{tab:resultsLognormalSum}
\end{table}

Furthermore, we wanted to empirically test the performance of the convolution method when employing it on a sum of independent Log-Normals that are not identically distributed. We therefore perform a experiment similar to the ones described in above, but instead estimate
\[\alpha = F_Y(11.2), \text{ with }Y = \sum_{i=1}^{16} X_i,\]
where $X_i \sim \text{Log-Normal}(0, \sigma_i)$ and $\sigma_i = \frac{1}{2^{2+j}}$ with $j = i \mod 4$ for $i \in \{1, 2, \hdots, 16\}$. The result is shown in the right plot of Figure \ref{fig:convLognorm}. From the graphs it is apparent that the convolution method performs well in this case as well.

\subsection{Estimating unknown distributions}
\label{sec:estUnknownDistr}
In this subsection we explore the convergence properties of the convolution method for Nakagami-m and the Rice distribution. The PDFs are given in Table \ref{table1}. Note that we let \(\Omega = 1\) for the Nakagami-m, while we for the Rice distribution use an alternative parameterization where we let \(K = \frac{\nu^2}{2}\) and \(\Omega = \nu^2 + 2\). Similarly to the Log-Normal distribution, which was the topic of the former subsection \ref{sec:Lognormal}, we do not know the exact distribution of a sum of i.i.d Nakagami-m or Rice RVs. We run the same experiment as before, utilizing the convolution method to estimate the value 
\[\alpha = F_Y(0.8) = \mathbb{P}(Y \leq 0.8) = \mathbb{P} \left (\sum_{i=1}^{16} X_i \leq 0.8 \right ),\]
with \(X_i\), $i=1,2,\cdots, 16$, are  i.i.d RVs drawn from either the Nakagami-m distribution or the Rice distribution. We first calculate a pseudo-reference solution \(\bar \alpha\) using a mesh consisting of \(N = 2^{20}\) intervals. We then calculate estimates \(\bar \alpha_{m_i}, i = 7, 8, \hdots, 15\) where we utilize a mesh of size \(N = 2^i\) in order to calculate \(\bar \alpha_{m_i}\). The relative error of the approximation relative to the pseudo-reference solution is then calculated by \[\delta_{m_i} = \frac{\abs{\bar \alpha - \bar \alpha_{m_i}}}{\bar \alpha}\]

For the Nakagami-m case we see from the left plot in Figure \ref{fig:l_nakagami_r_rice} that we end up with convergence rates of \(N^{-6}, N^{-4}\) and \(N^{-2}\) when choosing parameter values \(m=3, m=2\) and \(m=1\) respectively. These observations are consistent with Theorem \ref{thm:errorEstimate} as the pdf of the Nakagami-m distribution is zero at zero for all derivatives up to and including the forth derivative when \(m=3\), while the same is true up to the second derivative for \(m=2\). For the case \(m=1\) the first derivative is not zero at zero. 

The last distribution we will consider is the Rice distribution. The result from the numerical experiment is shown in the right plot of Figure \ref{fig:l_nakagami_r_rice}. Here the relative error more or less coincide for all tested parameter values. We also note that the empirically observed convergence rate is \(N^{-2}\). This also agree with the result from Theorem \ref{thm:errorEstimate} (see Remark \ref{rem:f_not_zero_at_zero}) as the first derivative of the pdf of the Rice distribution is not zero at zero. 
\begin{figure}
\center
  \includegraphics[width=0.49\linewidth]{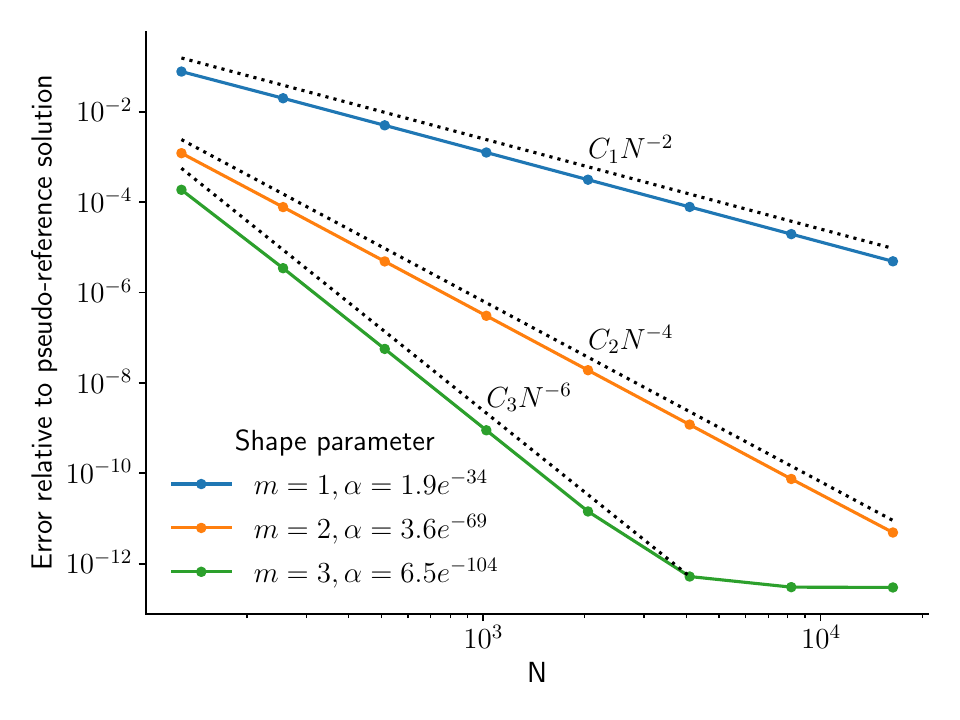}
  \includegraphics[width=0.49\linewidth]{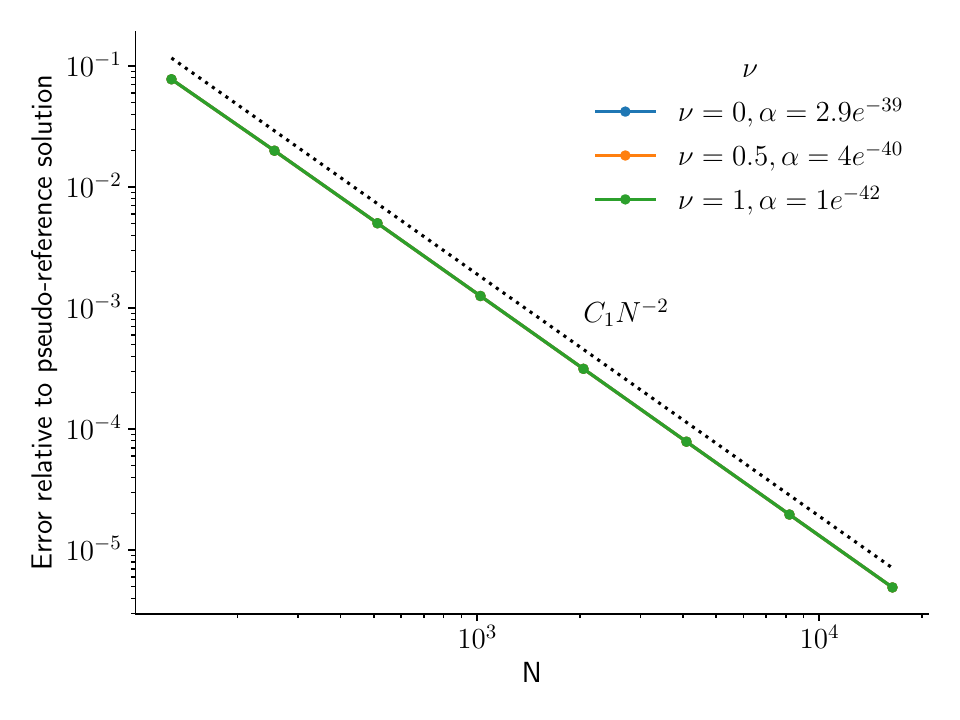}
  \caption{Left: Relative error as a function of the mesh-size when estimating $F_Y(0.8)$ where $Y = \sum_{i=1}^{16} X_i$ and $X_i \sim \text{Nakagami-m}(m)$. Right: Relative error as a function of the mesh-size when estimating $F_Y(0.8)$ where $Y = \sum_{i=1}^{16} X_i$ and $X_i \sim \text{Rice}(\nu)$}
  \label{fig:l_nakagami_r_rice}
\end{figure}

\section{Conclusion} \label{sec:conclusion}

We have presented a deterministic numerical method for estimating left-tail rare events of sums of non-negative independent RVs. The method is shown to be efficient, flexible, and accurate -- even when measured in relative error. This is due to the fact that numerical integration of convoluted densities only involves sums and products of non-negative floating point values, which are operations that are insensitive to rounding errors, cf.~Theorem~\ref{thm:errorWithRounding}. 
We further compare direct-based convolution 
to FFT-based convolution, and show by formal theoretical arguments and in numerical experiments that FFT-based convolution is more sensitive to rounding errors and a less reliable method when the magnitude of the probability of failure is sufficiently small. In the future, it would be interesting to explore whether ideas involving numerical integration of linear convolutions to could be extended to estimations of "left-tail" rare events for random vectors, and to right-tail rare events. 

\section*{Code availability}
The code used to run the numerical examples presented in this paper can be found at: \href{https://github.com/johannesvm/convolution-left-side-rare-events}{https://github.com/johannesvm/convolution-left-side-rare-events}

\bibliographystyle{amsplain}
\bibliography{References.bib}
\end{document}